\documentclass[format=sigconf, review=false]{acmart}
\usepackage{acm-arxiv}

\usepackage{balance}

\usepackage[ruled]{algorithm2e} 

\SetAlFnt{\small}
\SetAlCapFnt{\small}
\SetAlCapNameFnt{\small}
\SetAlCapHSkip{0pt}
\IncMargin{-\parindent}

\usepackage{mathtools}
\usepackage{verbatim}
\usepackage{subcaption}
\usepackage{bm}
\usepackage[bb=dsserif]{mathalpha} 
\usepackage [mathscr] {eucal}
\usepackage{float}
\usepackage{xcolor}
\usepackage[colorinlistoftodos]{todonotes}  
\usepackage{color}                          
\usepackage{totcount}                       
\usepackage{tikz}
\usetikzlibrary{intersections}
\usetikzlibrary{patterns} 
\usetikzlibrary{patterns.meta}
\usetikzlibrary{fit,positioning,backgrounds}
\usepackage{pgfplots,pgfplotstable}
\pgfplotsset{compat=1.17}
\usepgfplotslibrary{groupplots}
\usepgfplotslibrary{fillbetween}

\usepackage{cleveref}
\usepackage[shortlabels]{enumitem}

\newcommand{\cchi}[3]{\overset{\mathrm{c}}{\chi}{}^{#1, #2}_{#3}}

\makeatother
\definecolor{lightgreen}{rgb}{0.86, 0.93, 0.78}
\definecolor{bordergreen}{rgb}{0.55, 0.76, 0.74}
\definecolor{lightblue}{rgb}{0.70, 0.90, 0.99}
\definecolor{borderblue}{rgb}{0.01, 0.66, 0.96}
\definecolor{lightamber}{rgb}{1, 0.93, 0.70}
\definecolor{borderamber}{rgb}{1, 0.76, 0.03}
\definecolor{PaperBlue}{HTML}{1f77b4}
\definecolor{PaperOrange}{HTML}{ff7f0e}
\definecolor{PaperGreen}{HTML}{2ca02c}
\definecolor{PaperRed}{HTML}{e41a1c}      
\definecolor{PaperPurple}{HTML}{6a0dad}   
\definecolor{PaperPink}{HTML}{ff69b4}

\theoremstyle{acmplain}
\newtheorem{theorem}{Theorem}[section]
\newtheorem{definition}[theorem]{Definition}
\newtheorem{example}[theorem]{Example}
\newtheorem{property}[theorem]{Property}

\newtheorem{rthm}{Theorem}
\newenvironment{reptheorem}[1]
  {\renewcommand\therthm{\ref*{#1}}\rthm}
  {\endrthm}

\newtheorem{rlem}{Lemma}
\newenvironment{replemma}[1]
  {\renewcommand\therlem{\ref*{#1}}\rlem}
  {\endrlem}

\pgfplotstableread[col sep=comma]{
x,y,group,crw,unif
1.2728,2.9247,grass,0.09945,0.13043
3.6775,0.7866,campaign,0.09295,0.45652
}\tblFeedGroups
\pgfplotstableread[col sep=comma]{
x,y,group,w
0.5139,0.9472,far,0.02174
4.8076,2.3286,far,0.02174
0.5648,1.7325,far,0.02174
2.8743,0.6390,near,0.02174
4.4075,0.4547,near,0.02174
2.3474,2.0670,far,0.02174
4.4270,3.8251,far,0.02174
4.1773,1.1709,near,0.02174
0.0089,3.8938,far,0.02174
1.7904,1.2559,far,0.02174
2.8279,3.0931,far,0.02174
0.1821,2.8279,far,0.02174
2.2455,0.3634,far,0.02174
4.0968,3.2803,far,0.02174
5.2709,0.4093,far,0.02174
5.0986,1.5757,far,0.02174
3.3709,1.5986,near,0.02174
5.7861,2.4948,far,0.02174
3.6413,3.8822,far,0.02174
1.2728,2.9247,grass,0.13043
3.6775,0.7866,campaign,0.45652
}\tblFeedUnif
\pgfplotstableread[col sep=comma]{
x,y,group,w
0.5139,0.9472,far,0.04002
4.8076,2.3286,far,0.04076
0.5648,1.7325,far,0.04002
2.8743,0.6390,near,0.04377
4.4075,0.4547,near,0.04549
2.3474,2.0670,far,0.04557
4.4270,3.8251,far,0.03494
4.1773,1.1709,near,0.03767
0.0089,3.8938,far,0.04604
1.7904,1.2559,far,0.04557
2.8279,3.0931,far,0.04604
0.1821,2.8279,far,0.04604
2.2455,0.3634,far,0.04377
4.0968,3.2803,far,0.03494
5.2709,0.4093,far,0.04604
5.0986,1.5757,far,0.04076
3.3709,1.5986,near,0.04604
3.6775,0.7866,campaign,0.00377
5.7861,2.4948,far,0.04604
3.6413,3.8822,far,0.03805
1.3253,2.8492,grass,0.01589
1.3086,3.0386,grass,0.01375
1.2213,2.8532,grass,0.01589
1.2427,3.0109,grass,0.01422
1.3968,2.7442,grass,0.02026
1.1423,3.0523,grass,0.01943
3.6287,0.7925,campaign,0.00604
3.7125,0.8098,campaign,0.00538
3.6553,0.7799,campaign,0.00394
3.7100,0.7691,campaign,0.00462
3.6387,0.7704,campaign,0.00440
3.6770,0.7752,campaign,0.00379
3.7164,0.7731,campaign,0.00508
3.6941,0.7612,campaign,0.00425
3.6760,0.7482,campaign,0.00469
3.6933,0.7488,campaign,0.00390
3.6949,0.7496,campaign,0.00393
3.6766,0.8227,campaign,0.00421
3.6405,0.7908,campaign,0.00468
3.6652,0.8051,campaign,0.00418
3.6749,0.8216,campaign,0.00418
3.6615,0.7789,campaign,0.00389
3.6631,0.8181,campaign,0.00447
3.6506,0.7617,campaign,0.00448
3.6392,0.7648,campaign,0.00460
3.7023,0.8057,campaign,0.00444
}\tblFeedCrw
\pgfplotstableread[col sep=comma]{
n,unif,m0.25,lo0.25,hi0.25,m0.5,lo0.5,hi0.5,m1,lo1,hi1
1.000,3.333,3.344,2.635,3.570,3.424,2.477,3.889,3.397,2.443,4.046
2.000,6.452,3.853,3.543,4.143,3.830,3.349,4.311,3.863,3.143,4.317
4.000,12.121,5.242,4.573,5.902,4.663,4.142,5.352,4.584,3.992,5.315
11.000,27.500,9.587,8.494,10.583,6.998,6.365,8.271,6.102,5.336,7.200
31.000,51.667,18.278,17.155,19.552,11.412,10.668,12.883,8.551,7.669,10.018
101.000,77.692,29.501,28.750,30.622,17.059,16.422,18.493,11.492,10.536,13.409
301.000,91.212,36.431,35.830,37.468,20.532,19.879,22.148,13.296,12.227,15.150
1001.000,97.184,40.063,39.663,41.290,22.392,21.756,24.057,14.277,13.206,16.151
}\tblAttackSize
\pgfplotstableread[col sep=comma]{
dist,crw,unif,group
0.0000,-11.400,-3.226,campaign
0.0231,-16.303,-3.226,campaign
0.0369,-8.438,-3.226,campaign
0.0421,-7.674,-3.226,campaign
0.0491,-21.086,-3.226,campaign
0.6305,0.030,-3.226,near
0.8019,-0.194,-3.226,near
0.8166,-0.201,-3.226,near
0.8680,-0.192,-3.226,near
1.4932,-0.201,-3.226,far
1.6256,-0.216,-3.226,far
1.6375,-0.192,-3.226,far
1.8462,-0.194,-3.226,far
1.9119,-0.216,-3.226,far
1.9446,-0.194,-3.226,far
2.4580,-0.192,-3.226,far
2.5288,-0.219,-3.226,far
2.7138,-0.192,-3.226,far
3.0056,-0.334,-3.226,grass
3.0959,-0.232,-3.226,far
3.1284,-0.421,-3.226,grass
3.1296,-0.219,-3.226,far
3.1676,-0.220,-3.226,far
3.2099,-0.421,-3.226,grass
3.2532,-0.220,-3.226,far
3.2685,-0.376,-3.226,grass
3.2978,-0.365,-3.226,grass
3.4001,-0.347,-3.226,grass
4.0478,-0.192,-3.226,far
4.8076,-0.192,-3.226,far
}\tblMarginalVariant

\setcitestyle{authoryear}

\title{Breaking the Illusion of Artificial Consensus: Clone-Robust Weighting in Arbitrary Metric Spaces}

\author{Damien Berriaud}
\affiliation{%
  \institution{ETH Z\"urich}
  \city{Z\"urich}
  \country{Switzerland}
}
\email{dberriaud@ethz.ch}

\author{Supanat Wangsutthitham}
\affiliation{%
  \institution{ETH Z\"urich}
  \city{Z\"urich}
  \country{Switzerland}
}
\email{swangsutthit@ethz.ch}

\author{Roger Wattenhofer}
\affiliation{%
  \institution{ETH Z\"urich}
  \city{Z\"urich}
  \country{Switzerland}
}
\email{wattenhofer@ethz.ch}

\begin{abstract}
Social media, pseudonymous identities and generative AI have made online platforms highly vulnerable to Coordinated Inauthentic Behavior: by generating large numbers of similar posts, influence campaigns create an illusion of widespread support, exploiting the tendency of users and algorithms alike to mistake frequency for consensus.
Clone-robust weighting functions address this problem by assigning influence in a way that is insensitive to near-duplication, as measured by a metric, while remaining symmetric, continuous and local.
While prior constructions were restricted to Euclidean spaces, we present a computationally efficient framework for arbitrary metric spaces, based on neighborhood graphs and graph weighting functions satisfying a simple locality condition. As a consequence, the visibility an attacker gains from near-duplicates is bounded independently of their number.
We further argue that, in democratically sensitive settings such as news feeds, such guarantees must come with explainability, and introduce sharing coefficients that quantify how influence is transferred between elements of prior volumetric constructions.
In the graph setting, we show that sharing measures based on the deletion of vertices or edges are incompatible with locality, and instead define sharing through the addition of clones.
Finally, we characterize the resulting interpretability properties and construct clone-robust graph weighting functions satisfying them.
\end{abstract}

\keywords{Clone-Robustness, Graph Weighting Function, Metric Spaces, Redundancy}

\ccsdesc[500]{Theory of computation~Social networks}
\ccsdesc[300]{Mathematics of computing~Probabilistic representations}
\ccsdesc[300]{Mathematics of computing~Graph algorithms}
\ccsdesc[500]{Security and privacy~Social network security and privacy}

\begin{document}

\maketitle

\thispagestyle{empty}

\section{Introduction}

As you scroll through an online feed, you encounter a particular claim and dismiss it without much thought. Later, you see it again, phrased more cautiously; then as a joke, a meme, and a confident argument using technical language. Each instance looks different, comes from a different account, and appeals to a different audience. None is particularly convincing in isolation, yet the doubt slowly creeps in: is there some truth to this after all?

On online platforms, influence is implicitly tied to frequency: an opinion posted repeatedly tends to be perceived as more prevalent or credible than it truly is.
This popularity bias is amplified when repetition is manufactured, e.g., by paid accounts posting on command (troll farming), by campaigns disguised as spontaneous grassroots support (astroturfing), or more generally through what platforms call Coordinated Inauthentic Behavior (CIB). As a result, a small but hyper-active faction can hijack the conversation and drown out opposing opinions by saturating platforms with variations of the same unfounded claim.


If redundancy were a discrete phenomenon, this problem would be easy to address: exact duplicates could be detected automatically, merged, or made to share their influence.
Generative AI, however, entirely breaks this fragile assumption. Thousands of variations of a single claim can now be produced at negligible cost, differing in wording, tone or framing, and spanning a whole continuum of similarity, from near-identical paraphrases to loosely related arguments.
No threshold can then reliably separate duplicates from independent content: redundancy is not a binary property of pairs of messages, but a matter of degree, calling for a metric perspective.

This shift, however, raises a new challenge: how should influence be assigned so that approximate clones share it locally, without allowing dense clusters of near-copies to overwhelm distant and unrelated opinions?
This question is not specific to online platforms: it arises whenever redundant inputs are aggregated, e.g., when evaluating models on benchmarks or aggregating human preferences in reinforcement learning (Section~\ref{sec:related_work}).
Clone-robust weighting in metric spaces provides a principled framework to address it, which we study here with CIB as our leading example.

\paragraph{Contribution \& Outline.}

We first propose to determine the probability of appearance of messages through clone-robust weighting functions, which assign positive weights, respect the symmetries of the underlying space, vary continuously, and only redistribute weight locally when near-duplicates are added (Section~\ref{sec:def_clone_robust_weighting}).
To construct such functions, we recover a binary notion of similarity from the metric: for each distance threshold, we connect sufficiently similar elements in a neighborhood graph, in which approximate clones are equivalent for most thresholds (Section~\ref{sec:general_weighting function}).
Any \emph{graph weighting function} satisfying symmetry and locality then induces, by aggregation across radii, a clone-robust weighting function that is efficiently computable and applies to arbitrary metrics, independently of the topology they induce (Theorem~\ref{thm:general_weight_function}).
As a consequence, the visibility gained by injecting near-duplicates of a post is bounded independently of their number (Corollary~\ref{cor:attack_budget}).
We further argue that robustness requires explainability, and introduce \emph{sharing coefficients} for prior clone-robust constructions, quantifying how influence is transferred between elements (Section~\ref{sec:sharing_measures_gr}).
Seeking similar measures for graph weighting functions, we show that sharing based on the deletion of vertices or edges is incompatible with locality (Theorems~\ref{thm:impossible_vertex_sharing} and~\ref{thm:impossible_edge_sharing}), and instead define sharing through the addition of clones (Section~\ref{sec:graph_weighting_design}).
We characterize how the resulting interpretability properties relate to each other and to locality, and provide a general construction of clone-robust graph weighting functions satisfying them with optimal locality (Theorem~\ref{thm:satisfyability_clone_sharing}).

\section{Related Work}\label{sec:related_work}

\paragraph{Measuring CIB on Online Platforms.}
Coordinated inauthentic behavior has become a central concern for social media platforms, from coordinated accounts fabricating trending topics and deleting their activity to evade detection~\cite{elmas_ephemeral_2021}, to botnets posting LLM-generated content that state-of-the-art text classifiers cannot distinguish from human posts~\cite{yang_anatomy_2024}.
Since individual items are increasingly hard to flag in isolation, the dominant detection paradigm instead measures \emph{similarity between accounts}~\cite{mannocci2026detectioncharacterizationcoordinatedonline}: accounts are linked through shared behavioral traces such as co-retweets, URLs, hashtag sequences or images, weighted by cosine or TF-IDF similarity~\cite{pacheco_uncovering_2021,nizzoli_coordinated_2021}, through embeddings of their textual content~\cite{luceri_unmasking_2024}, through the similarity of their action sequences~\cite{cresci_dna-inspired_2016} or activity time series~\cite{chavoshi_debot_2016}, and, on video-first platforms, through reused captions and media~\cite{luceri_coordinated_2026}.
These similarity spaces provide precisely the input required by our framework.
Rather than thresholding them to classify and remove suspected accounts, an operation that is brittle to the choice of threshold and to borderline cases, clone-robust weighting functions continuously discount near-duplicate content in proportion to its redundancy, so that platforms can mitigate the amplification effect of CIB campaigns without having to label any individual account as inauthentic.
Closest in spirit, \citet{nizzoli_coordinated_2021} quantify coordination as a matter of degree, by tracking which communities survive as a similarity threshold is progressively raised; our framework instead aggregates over all thresholds at once, through the measure $\nu$, and uses the result to reweight content rather than to label communities.

\paragraph{Reweighting Redundant Data}
In machine learning, training data rarely reflect the distribution of interest, and a standard remedy is to reweight samples in the loss, e.g., according to their difficulty~\cite{freund_decision-theoretic_1997,jiang_self-paced_2015} or through meta-learning on a small unbiased validation set~\cite{ren_learning_2019,shu_meta-weight-net_2019}.
Closest to our setting is the problem of redundancy among training samples, as arises in imbalanced datasets.
These are commonly handled by weighting samples inversely to their class frequency~\cite{lin2017focal,cao_learning_2019,jamal_rethinking_2020}, which amounts to the class-uniform weighting $w^{\mathrm{CU}}$ when redundancy is known exactly.
\citet{cui_class-balanced_2019} refine this by accounting for overlap between samples through an \emph{effective number of samples}, so that each sample only contributes the additional coverage it provides, an idea closely related to the private weight of Section~\ref{sec:sharing_measures_gr}.

\paragraph{Clone-Robust Aggregation and Evaluation}
Robustness to clones has a long history in social choice, both for alternatives---through independence of clones~\cite{tideman_independence_1987} and composition consistency~\cite{laffond_composition-consistent_1996,brandl_consistent_2016}---and for agents, through false-name-proofness~\cite{conitzer_using_2010,nehama_manipulation-resistant_2022,todo_characterizing_2009}.
It has recently resurfaced in the evaluation of machine learning models, where benchmarking is cast as a game between benchmarks and models and independence of clones is sought through suitable aggregation mechanisms~\cite{balduzzi_re-evaluating_2018,gao_re-evaluating_2025,marris_deviation_2025,liu_re-evaluating_2025}.
These works, however, mostly consider exact clones.
Robustness to \emph{approximate} clones was concurrently introduced by~\cite{procaccia_clone-robust_2025} for preference aggregation in RLHF, weighting alternatives by the volume of their Voronoi cells, and formalized axiomatically by~\cite{berriaud_clone-robust_2025} for weighting functions in Euclidean spaces.
The latter rests on four principles:
\emph{positivity} (every element receives positive weight),
\emph{symmetry} (weights respect the symmetries of the space),
\emph{continuity} (weights vary smoothly under small perturbations), and
\emph{locality} (adding an element only affects the weights of nearby ones).
It also highlighted the limitations of Voronoi-based weights and introduced geometrically grounded alternatives.
The present work extends this axiomatic framework to arbitrary metric spaces and complements it with sharing coefficients that make the resulting weights interpretable.

\section{Clone-Robust Weighting Functions}\label{sec:def_clone_robust_weighting}

Consider an algorithmic news aggregator that collects posts from multiple media outlets and users into a centralized feed. When a user refreshes their feed, the system gathers all newly available posts into a finite set $S\subseteq M$, where $M$ denotes the space of all possible posts.
The algorithm then selects a probability distribution $\pi_S$ over $S$ and populates the user's feed by independently\footnote{We assume that the probability of sampling the same post twice is negligible, i.e., $K^2 \Vert \pi_S\Vert_2^2\ll1$, where $K$ denotes the total number of posts in the feed.}
sampling posts according to $\pi_S.$

The space of posts $M$ is equipped with a notion of distance, given by an operator
$d:M\times M \mapsto \mathbb{R}_{\geq 0}$,
which satisfies \emph{identity}, i.e., $d(x,x) =0$, \emph{symmetry}, i.e., $d(x,y) = d(y,x)$, and \emph{triangle inequality}, i.e., $d(x,z) \leq d(x,y)+d(y,z)$, for all $x,y,z\in M.$
Note that $d$ is in general a pseudo-metric rather than a metric, as it does not satisfy separability: distinct posts $x\neq y$ may satisfy $d(x,y) =0$, for instance when two identical posts originate from different accounts.
Since $d$ determines what is considered a duplication in the first place, and media revenue may depend on post visibility, it should be chosen in a transparent and agreed-upon manner.
Our framework leaves this choice entirely open: $d$ may be derived from any existing measure of similarity between posts or accounts~\cite{pacheco_uncovering_2021,luceri_unmasking_2024,luceri_coordinated_2026}, and refined independently of the weighting itself.
We thus take the pseudo-metric $d$ as given, and study how it should shape the distribution $\pi_S$.

Formally, we introduce a \emph{weighting function $f$}, which assigns to each finite subset $S\subseteq M$ a probability distribution over its elements, 
\begin{equation*}
    f: S\in \bigcup_{n\geq1} \mathcal{P}_n(M)\mapsto \pi_S\in \Delta(S),
\end{equation*}
where $ \mathcal{P}_n(M)$ denotes the family of all subsets of $M$ with cardinality $n\in\mathbb{N}$, and $\Delta(S) = \big\{ \pi_S : S \mapsto [0,1] \mid \sum_{x \in S} \pi_S(x) =1 \big\}$ denotes the simplex over the elements of $S.$

We next discuss a collection of axioms capturing desirable properties for such weighting functions, each motivated by a way in which weights could otherwise be manipulated.
First, weights should respect the symmetries of the underlying pseudo-metric space: posts that are isomorphic within $S$ should receive equal weight.
Second, since perfect clones (i.e., distinct posts at distance zero) are always isomorphic, we extend this requirement to approximate clones and ask that nearby posts receive similar weights. Moreover, the weighting should be robust to perturbations of the set $S$ itself, i.e., $f$ should be continuous. Together, they prevent weights from changing abruptly under small variations of the posts, which generative AI easily produces.
Third, we require protection against duplication attacks. When approximate copies of a post are added, the weight of posts farther than a disambiguation threshold $\alpha>0$ should remain unaffected, so that weight is only redistributed locally. The threshold $\alpha$ thus sets the scale beyond which posts are considered unrelated, and this locality is what bounds the visibility an attacker can gain by duplication (Corollary~\ref{cor:attack_budget}).
Finally, we impose a non-censorship requirement: every post must receive strictly positive weight, simultaneously ensuring that an attacker cannot drive the weight of a post to zero through duplication.
Note that these axioms are agnostic to the quality or relevance of posts and only address redundancy: promoting the ``best'' among clones would be the role of a downstream ranking layer, combining its relevance scores with the clone-robust weights.
We now provide a formal definition of these desiderata, with the regularity properties expressed via Lipschitz-like conditions.

\begin{definition}[Clone-Robust Weighting Function]\label{def:axi_clone_robust}
Let $\alpha>0$ be a disambiguation factor, let $n\in \mathbb{N}_{>0}$, and let $L_n$, $L_n'$, $L_n''>0$ be Lipschitz constants.
A weighting function $f$ is said to satisfy a given property at cardinality $n$ if the corresponding condition holds for every subset $S\in\mathcal{P}_n(M)$, every $x,y \in S$, every self-isometry $\sigma:S\mapsto S$ and every finite subset $Z\subseteq M\setminus S.$

\begin{itemize}
    \item \textbf{Positivity}, i.e., $f(S)(x) >0.$ 
    
    \item \textbf{Symmetry}, i.e.,  $f(S)(x) = f(S)(\sigma(x)).$ 
        
    \item \textbf{Lipschitz Clone Fairness}, i.e., $ \vert f (S)(x) - f(S)(y) \vert \leq L_n d(x,y).$ 

    \item \textbf{Lipschitz Continuity}, i.e., $\vert f (S)(x) - f (\pi(S))(\pi(x)) \vert \leq L_n' ~\max_{u\in S} d(u,\pi(u))$, for each bijection $\pi:S\mapsto \pi(S) \subseteq M.$ 
        
    \item $\bm{\alpha}\textbf{-Lipschitz Locality}$, i.e., if  $d(x,y)\geq \alpha $, then 
    $ \vert f (S \cup Z)(y) - f(S)(y) \vert \leq L_n''~\max_{z\in Z} d(x,z).$

\end{itemize}

\noindent A weighting function satisfying the above properties for every $n\in\mathbb{N}_{>0}$ is called $\boldsymbol{\alpha}$\textbf{-clone-robust}. 

\end{definition}

\noindent Note that the above properties strengthen the axioms proposed in~\cite{berriaud_clone-robust_2025} by providing quantitative Lipschitz control---where differences are explicitly bounded in proportion to distance---rather than the qualitative uniform continuity-type guarantees of the form ``for all $\varepsilon>0$, there exists $\delta>0\dots$.'' 
Moreover, the present locality property generalizes the original formulation in~\cite{berriaud_clone-robust_2025} by considering the addition of arbitrary finite subset $Z$ instead of singletons only.

\section{Construction based on Neighborhood Graphs}\label{sec:general_weighting function}

A central difficulty in clone-robust weighting is that similarity between posts varies gradually with distance, while local sharing of influence is easily defined only when similarity is all-or-nothing, as captured by an equivalence relation.
To reconcile both perspectives, we threshold the metric to recover a binary notion of closeness.

Let $S \subseteq M$ be finite and $r\geq 0$ be a radius.
We define the $r$-neighborhood graph $G_r(S) = \big( S, E_r(S)\big)$ as the undirected graph on $S$ with edge set
$E_r(S) = \{ (x,y) \in S^2 \mid d(x,y) \leq r\}$.
Adjacency in $G_r(S)$ is a similarity relation that is generally not transitive: $x \sim y$ and $y \sim z$ do not imply $x \sim z$, although the triangle inequality gives $d(x,z) \leq 2r$, i.e., $x \sim z$ in $G_{2r}(S)$.
To recover a transitive notion of similarity, we consider the closed neighborhood $N_G[x] = \{x\} \cup \{ y \in V \mid (x,y) \in E \}$ of a node $x$ in a graph $G =(V,E)$.
It induces an equivalence relation on $V$, where $x \equiv_G y$ if and only if $N_G[x] = N_G[y]$, and we denote by $[x]_G$ the equivalence class of $x$.
We further write $d_G$ for the shortest-path distance in $G$, with $d_G(x,y)= \infty$ if $x$ and $y$ are not connected.

Perfect clones in $S$, i.e., elements $x,y \in S$ with $d(x,y)=0$, are naturally always equivalent in $G_r(S).$ 
More importantly, approximate clones are also equivalent in $G_r(S)$ for ``most values of $r$,'' as illustrated in Figure~\ref{fig:forbidden_intervals_approximate_clones}.
Hence, the graphs $(G_r(S))_{r\geq0}$ and their equivalence classes are essentially invariant under the addition of clones, motivating the notion of \emph{clone-robust graph weighting function}.

\begin{figure*}[htb]
    \centering
    \begin{minipage}[c]{0.49\textwidth}
         \centering
         \begin{tikzpicture}[scale=0.9]

              \draw[->] (0,0) -- (8,0) node[right] {$r$ };
              \node[above] at (0,0.1) {0};
              \draw[-] (0,-0.07) -- (0,0.07); 
            
              \def\dxy{0.4} 
              \node[circle, fill = black, inner sep=1pt] at (\dxy,0){};
              \node[below,font=\small] at (\dxy,0) {$d(x,\textcolor{PaperBlue}{y})$};
              \fill[PaperBlue, fill opacity=0.25] (0,0) rectangle (2* \dxy, 0.1);

              \def\poslist{1.9, 3.3, 5.3, 7.1}
              \foreach \pos [count=\i] in \poslist {
                \pgfmathparse{\i-1}
                \pgfmathtruncatemacro{\ci}{mod(int(\pgfmathresult),5)}
                \pgfkeysgetvalue{/tikz/pgfmathresult}{\pgfmathresult}
                \ifnum\ci=0 \def\thiscolor{PaperOrange}\fi
                \ifnum\ci=1 \def\thiscolor{PaperGreen}\fi
                \ifnum\ci=2 \def\thiscolor{PaperRed}\fi
                \ifnum\ci=3 \def\thiscolor{PaperPurple}\fi
                \node[below,font=\small] at (\pos,0) {$d(x,\textcolor{\thiscolor}{z_{\i}})$};
                \node[circle, fill = \thiscolor, inner sep=1pt] at (\pos,0){};
                \pgfmathsetmacro{\leftx}{\pos - \dxy}
                \pgfmathsetmacro{\rightx}{\pos + \dxy}
                \fill[PaperBlue, fill opacity=0.25] (\leftx,0) rectangle (\rightx,0.1);
              }
            
              \def\poslist{1.55, 2.92, 5.54, 7.2}
              \foreach \pos [count=\i] in \poslist {
                \pgfmathparse{\i-1}
                \pgfmathtruncatemacro{\ci}{mod(int(\pgfmathresult),5)}
                \pgfkeysgetvalue{/tikz/pgfmathresult}{\pgfmathresult}
                \ifnum\ci=0 \def\thiscolor{PaperOrange}\fi
                \ifnum\ci=1 \def\thiscolor{PaperGreen}\fi
                \ifnum\ci=2 \def\thiscolor{PaperRed}\fi
                \ifnum\ci=3 \def\thiscolor{PaperPurple}\fi
                \node[above, font=\small] at (\pos,0.05) {$d(\textcolor{PaperBlue}{y},\textcolor{\thiscolor}{z_{\i}})$};
                \node[circle, fill = PaperBlue, inner sep=1pt] at (\pos,0){};
              }
            
            \end{tikzpicture}

    \end{minipage}
    \hfill
    \begin{minipage}[c]{0.49\textwidth}
        \centering
        \begin{tikzpicture}[scale=0.86,every node/.style={circle, thick, draw, minimum size=8mm, inner sep=0pt, font=\small}]

            \node[draw=PaperBlue, fill=PaperBlue!4, ] (0) at (0,1) {0}; 
            \node[draw=PaperBlue, fill=PaperBlue!4, ] (1) at (0,-1) {1};
            \node[draw=PaperBlue, fill=PaperBlue!4] (2) at (1.732,0) {2};
            \node[draw=PaperOrange, fill=PaperOrange!4] (3) at (3.732,0) {3};
            \node[draw=PaperGreen, fill=PaperGreen!4] (4) at (5.732, 1) {4};
            \node[draw=PaperGreen, fill=PaperGreen!4] (5) at (5.732, -1) {5};
            \node[draw=PaperPink, fill=PaperRed!4] (6) at (7.732, 1) {6};
            \node[draw=PaperPurple, fill=PaperPurple!4] (7) at (7.732, -1) {7};
            
            \node[draw=none, circle=none, font=\normalsize\bfseries, text = PaperRed] at (3.1, -1.2) {$ \mathbf{\big\vert V/{\equiv_G} \big\vert = 5}$ };
        
            \node[draw=none, circle=none, font=\normalsize] at (0.7, 1.4) {$\frac{1}{\textcolor{PaperRed}{5} \cdot \textcolor{PaperBlue}{3}}$}; 
            \node[draw=none, circle=none, font=\normalsize] at (0.7, -1.3) {$\frac{1}{\textcolor{PaperRed}{5} \cdot \textcolor{PaperBlue}{3}}$}; 
            \node[draw=none, circle=none, font=\normalsize] at (1.9, 0.8) {$\frac{1}{\textcolor{PaperRed}{5} \cdot \textcolor{PaperBlue}{3}}$}; 
            \node[draw=none, circle=none, font=\normalsize] at (3.732, 0.8) {$\frac{1}{\textcolor{PaperRed}{5} \cdot \textcolor{PaperOrange}{1}}$}; 
            \node[draw=none, circle=none, font=\normalsize] at (5.032, 1.4) {$\frac{1}{\textcolor{PaperRed}{5} \cdot \textcolor{PaperGreen}{2}}$}; 
            \node[draw=none, circle=none, font=\normalsize] at (5.032, -1.3) {$\frac{1}{\textcolor{PaperRed}{5} \cdot \textcolor{PaperGreen}{2}}$}; 
            \node[draw=none, circle=none, font=\normalsize] at (8.432, 1.4) {$\frac{1}{\textcolor{PaperRed}{5} \cdot \textcolor{PaperPink}{1}}$}; 
            \node[draw=none, circle=none, font=\normalsize] at (8.432, -1.3) {$\frac{1}{\textcolor{PaperRed}{5} \cdot \textcolor{PaperPurple}{1}}$}; 
            
          \draw (0) -- (1) -- (2) -- (0) -- (3) -- (1);
          \draw (2) -- (3)-- (4)-- (5) -- (3);
          \draw (4) -- (6) -- (5);
          \draw (4) -- (7) -- (5);
        
        \end{tikzpicture}
    \end{minipage}

    \vspace{-3pt} 
   
    \begin{minipage}[t]{0.47\textwidth}
        \caption{
        Impact of adding a vertex $y$ that is an approximate clone of $x$.
        By the triangle inequality, each distance $d(y,z)$ with $z\in S$ lies within an interval of length $2d(x,y)$ centered at $d(x,z)$.
        Thus, for all radii $r \ge 0$ outside the blue forbidden intervals, the vertices $x$ and $y$ are equivalent in $G_r(S \cup \{y\})$.
        }
        \label{fig:forbidden_intervals_approximate_clones}
    \end{minipage}
    \hfill
    \begin{minipage}[t]{0.47\textwidth}
        \caption{Illustration of the class-uniform weighting $w^{\mathrm{CU}}$ on a small graph. 
        Nodes are grouped into $\vert V/{\equiv_G}\vert =5$ equivalence classes: 
        $\{0,1,2\}$ in blue, $\{3\}$ in orange, $\{4,5\}$ in green, $\{6\}$ in pink, and $\{7\}$ in purple. 
        Each class receives the same total weight, distributed uniformly among its vertices.}
        \label{fig:w^CU}
    \end{minipage}
    \Description{Left: a horizontal axis of radii on which the distances from a post x to the other posts are marked; around each of these distances, a shaded blue interval of width twice the distance between x and its approximate clone y is drawn. Outside these intervals, x and y have the same neighbors. Right: a graph with eight vertices, numbered 0 to 7. Vertices 0, 1 and 2 form a triangle and are all connected to vertex 3; vertex 3 is connected to vertices 4 and 5, which are connected to each other; vertices 6 and 7 are each connected to both 4 and 5. The vertices are colored by equivalence class: {0, 1, 2}, {3}, {4, 5}, {6} and {7}, i.e., five classes. Each vertex is labeled with its weight, one fifth divided by the size of its class: one fifteenth for vertices 0 to 2, one fifth for vertex 3, one tenth for vertices 4 and 5, and one fifth for vertices 6 and 7.}
    
\end{figure*}

\begin{definition}[Clone-Robust Graph Weighting Functions] \label{def:clone_rob_graph}
A graph weighting function $w$ associates with each finite graph a probability distribution over its vertices, i.e.,
 $w: (V,E)\in \mathcal{G} \mapsto p_V \in \Delta(V),$
where $\mathcal{G}$ denotes the set containing all finite undirected unweighted graphs.

\noindent Let $k >0$ be a positive disambiguation factor. 
We say that a graph weighting function $w$ satisfies a given property if it verifies the corresponding condition for every graph $G\in\mathcal{G}$ and vertex $x\in V(G).$ 

\begin{itemize}

\item \textbf{Symmetry}, i.e.,  $w(G)(x) = w(\sigma(G))(\sigma(x))$ for every graph isomorphism $\sigma: \mathcal{G}\mapsto \mathcal{G}.$ 

\item $\bm{k}\textbf{-Locality}$, i.e.,  $w(G)(z) = w(G \setminus \{y\})(z)$ 
for each pair of vertices $y,z \in V(G)$ where $y \equiv_G x$, $y \neq x$, and $d_G(x,z) > k.$

\end{itemize}

\noindent A graph weighting function satisfying the above properties is called $\bm{k}\textbf{-clone-robust}$.
\end{definition}

We next construct a simple instance of $1$-clone-robust graph weighting function.

\begin{example}\label{exple:w^CU}
For a finite undirected graph $G=(V,E)$, the equivalence relation $\equiv_G$ partitions $V$ into classes of structurally identical vertices.  
A natural way to assign weights in such a graph is to let each class account for the same fraction of the distribution, and treat all vertices within a class equally (cf. Figure~\ref{fig:w^CU}).
We call the resulting rule the \textbf{class-uniform graph weighting function}.  
Formally, if $V/{\equiv_G}$ denotes the set of equivalence classes in $G$, we set
$$ w^{\mathrm{CU}}(G)(x) \; \coloneq \; \frac{1}{\vert V/{\equiv_G}\vert \cdot \vert [x]_G\vert}, \qquad x \in V.$$

\noindent The function $w^{\mathrm{CU}}$ satisfies the required properties:
symmetry holds because graph isomorphisms preserve equivalence classes, and 1-locality follows from the fact that adding a new vertex to an existing equivalence class affects only the weights of vertices within that class.

\end{example}

Clone-robust graph weighting functions are of special interest because they enable the construction of clone-robust weighting functions for arbitrary pseudo-metric.

\begin{theorem}\label{thm:general_weight_function}
    Let $w$ be a $k$-clone-robust graph weighting function, $\alpha$ be a positive radius and $\nu$ be a probability density function with support $[0,\alpha]$ and value bounded from above by $\overline{\nu}.$
    Then the weighting function $f_{\nu,w}$  is  $k \cdot \alpha$-clone-robust, where $f_{\nu,w}$ is defined for every subset $S\subseteq M$ of cardinality $\vert S\vert=m \in \mathbb{N}_{>0}$ and $x\in S$ by
    
    $$f_{\nu,w}(S)(x) \;:=\; \int_0^\alpha \nu(r)\, w(G_r(S))(x)\,dr.$$
   
    \noindent Moreover, $f_{\nu,w}(S)(x)$ can be computed in time $O\big( m^2 (W(m) + C_\nu + C_d + \log m )\big)$, where $W(m)$, $C_\nu$ and $C_d$ are the worst-case costs of evaluating $w$ on a graph with $m$ vertices, the cumulative distribution of $\nu$ at a single argument, and the pseudo-metric $d$.
\end{theorem}

A precise formulation of Theorem~\ref{thm:general_weight_function} providing  the exact Lipschitz constants, its full proof, as well as an algorithm  achieving the stated complexity, are detailed in Appendix~\ref{app:proof_main}. 
Interestingly, note that positivity is not required from graph weighting function $w$; 
this property rather emerges for $f_{\nu,w}$ from symmetry and locality alone at neighborhood graphs of radii $r$ close to 0. 
Furthermore, the strict equality at the singleton level in the locality condition for graphs (Definition~\ref{def:clone_rob_graph}) can be applied iteratively, which yields the locality property of $f_{\nu,w}$ for arbitrary finite set $Z$ (Definition~\ref{def:axi_clone_robust}).

Besides its Lipschitz-style guarantees, this graph-based approach offers  two crucial advantages over prior constructions~\cite{berriaud_clone-robust_2025}.
First, the exact evaluation of $f_{\nu,w}$ remains computationally tractable under suitable choices of $\nu$ and polynomial-time graph weighting functions like $w^{CU}.$
Second, this formulation is topologically invariant, and the weighting procedure depends only on the underlying metric $d$ rather than the specific choice of embedding space $M$. 
As such, it remains unaffected by isometric embeddings---such as the addition of dummy dimensions when $M = \mathbb{R}^n$---and offers a uniform framework across all metric spaces instead of being restricted to Euclidean settings.

\paragraph{Implications for Online News Feeds.}
Recall that the feed of a user is populated by sampling posts according to $\pi_S = f(S)$.
Under uniform weighting, an attacker injecting $\ell$ similar messages into a feed of $\vert S\vert$ posts obtains a share $\ell/(\vert S\vert + \ell)$, which tends to one as $\ell$ grows.
In contrast, Theorem~\ref{thm:general_weight_function} bounds the visibility of such a campaign independently of its size.

\begin{corollary}\label{cor:attack_budget}
    Let $f_{\nu,w}$ be as in Theorem~\ref{thm:general_weight_function}, and $S \subseteq M$ be a finite set of posts.
    Consider an attacker who promotes a target message $t \in S$ by injecting a finite set of variants $A \subseteq M \setminus S$, all within distance $\rho$ of $t$.
    Then the total visibility of the campaign satisfies
    \begin{equation*}
        \sum_{a \in A\cup\{t\}} f_{\nu,w}(S \cup A)(a)
        \;\leq \sum_{\substack{y \in S\\ d(t,y) < k\alpha}} f_{\nu,w}(S)(y) \;+\; 2\,\overline{\nu}\, \vert S \vert^2 \rho.
    \end{equation*}
\end{corollary}

\begin{proof}
    Every $y \in S$ with $d(t,y) \geq k\alpha$ satisfies $f_{\nu,w}(S\cup A)(y) \geq f_{\nu,w}(S)(y) - 2\overline{\nu}\vert S\vert\rho$ by Lipschitz $(k\alpha)$-locality, applied with $x = t$ and $Z = A$.
    Summing over these at most $\vert S\vert$ elements, using that $f_{\nu,w}(S\cup A)$ and $f_{\nu,w}(S)$ both sum to one, and dropping the non-negative weights of the remaining items within distance $k\alpha$ of $t$ yields the claim.
\end{proof}

\noindent The bound does not depend on the number $\vert A\vert$ of injected variants, only on their spread $\rho$: for $\nu$ uniform, i.e., $\overline{\nu} = 1/\alpha$, the second term scales with $\rho/\alpha$.
Since the bound holds for every $\rho$, an attacker can only escape it by spreading content further apart, i.e., by producing genuinely dissimilar messages.
Conversely, a defender who estimates the distance $\rho$ beyond which content no longer reinforces the same claim can choose $\alpha$ to minimize the right-hand side, balancing the weight of the $k\alpha$-neighborhood of $t$ against the spread term.

As an illustrative toy experiment, we consider a feed of $26$ posts in the plane, including a grassroots cluster of $6$ similar posts from independent users, into which an attacker injects $20$ close variants of a target post. 
Under uniform weighting, the campaign captures close to half of the feed, whereas clone-robust weighting reduces it to a share comparable to that of a few independent posts (Figure~\ref{fig:illustration}a).
Note though that the grassroots cluster is discounted as well: if coordinated and genuine agreement cannot be told apart in the metric space, any protection against redundancy necessarily affects both.
We then vary the size of the campaign in random feeds of $30$ posts: whereas its share tends to one under uniform weighting, it saturates under clone-robust weighting, at a level roughly proportional to $\rho/\alpha$ (Figure~\ref{fig:illustration}b).
This matches the second term of Corollary~\ref{cor:attack_budget}: below radius roughly $2\rho$, i.e., on a fraction $2\rho/\alpha$ of the integration range, the variants are not yet equivalent and capture most of the weight, whereas they form a single class at larger radii.
The first term is pessimistic here, as neighbors only yield weight to the campaign once equivalent to it, i.e., as $\alpha$ approaches the diameter of the feed.

\begin{figure*}[tbh]
    \centering
    \def\panelwidthval{4}   
    \def\panelheightval{2.75}  
    \def\panelsepval{1.35}    
    \edef\panelwidth{\panelwidthval cm}\edef\panelheight{\panelheightval cm}\edef\panelsep{\panelsepval cm}
    \pgfplotsset{
      illu/.style={
        scale only axis, width=\panelwidth, height=\panelheight,
        tick label style={font=\scriptsize},
        label style={font=\scriptsize},
        title style={font=\small, yshift=-1ex},
        legend style={font=\scriptsize, draw=none, fill=white, fill opacity=0.85, text opacity=1, inner sep=1pt, row sep=-1pt},
        legend cell align=left,
      },
    }
    \begin{tikzpicture}
    \begin{groupplot}[group style={group size=3 by 1, horizontal sep=\panelsep}, illu]
    
    \nextgroupplot[
      title={(a) Feed under attack (20 variants)},
      xmin=-0.3, xmax=6.3, ymin=-0.6, ymax={-0.6+6.6*\panelheightval/\panelwidthval},
      xtick=\empty, ytick=\empty,
      xlabel style={yshift=-2.56ex},
      xlabel={Total share: uniform $\to$ clone-robust},
      legend to name=legendA, legend columns=1,
      legend style={font=\scriptsize, draw=none, row sep=1pt},
      wclass/.style={scatter, only marks, scatter src=explicit symbolic, forget plot,
        scatter/classes={
          far={mark=*, draw=PaperBlue!80!black, fill=PaperBlue},
          near={mark=*, draw=PaperOrange!80!black, fill=PaperOrange},
          grass={mark=*, draw=PaperGreen!80!black, fill=PaperGreen},
          campaign={mark=*, draw=PaperRed!80!black, fill=PaperRed}},
        every mark/.append style={fill opacity=0.3},
        visualization depends on={14*sqrt(\thisrow{w}) \as \msize},
        scatter/@pre marker code/.append style={/tikz/mark size=\msize}},
    ]
    \pgfplotstablecopy{\tblFeedGroups}\to\grp
    \pgfplotstablegetelem{0}{x}\of{\grp}\let\gx\pgfplotsretval
    \pgfplotstablegetelem{0}{y}\of{\grp}\let\gy\pgfplotsretval
    \pgfplotstablegetelem{0}{crw}\of{\grp}\let\gc\pgfplotsretval
    \pgfplotstablegetelem{0}{unif}\of{\grp}\let\gu\pgfplotsretval
    \pgfplotstablegetelem{1}{x}\of{\grp}\let\tx\pgfplotsretval
    \pgfplotstablegetelem{1}{y}\of{\grp}\let\ty\pgfplotsretval
    \pgfplotstablegetelem{1}{crw}\of{\grp}\let\cc\pgfplotsretval
    \pgfplotstablegetelem{1}{unif}\of{\grp}\let\cu\pgfplotsretval
    \addplot[wclass, every mark/.append style={densely dotted}] table[x=x, y=y, meta=group, col sep=comma] {\tblFeedUnif};
    \addplot[wclass] table[x=x, y=y, meta=group, col sep=comma] {\tblFeedCrw};
    \pgfmathparse{100*\gu}\pgfmathprintnumberto[fixed,precision=0]{\pgfmathresult}{\guS}\global\let\guS\guS
    \pgfmathparse{100*\gc}\pgfmathprintnumberto[fixed,precision=0]{\pgfmathresult}{\gcS}\global\let\gcS\gcS
    \pgfmathparse{100*\cu}\pgfmathprintnumberto[fixed,precision=0]{\pgfmathresult}{\cuS}\global\let\cuS\cuS
    \pgfmathparse{100*\cc}\pgfmathprintnumberto[fixed,precision=0]{\pgfmathresult}{\ccS}\global\let\ccS\ccS
    \node[font=\scriptsize, text=PaperGreen!60!black, inner sep=1pt, anchor=north]
      at (axis cs:\gx,{\gy+0.8}) {\guS\,\% $\to$ \gcS\,\%};
    \node[font=\scriptsize, text=PaperRed!80!black, inner sep=1pt, anchor=north]
      at (axis cs:\tx,-0.15) {\cuS\,\% $\to$ \ccS\,\%};
    \draw[black!60, dashdotted] (axis cs:\tx,\ty) circle[radius=1];
    \def\arrowangle{70}
    \draw[->, black!70, thin] (axis cs:\tx,\ty) -- (axis cs:{\tx+cos(\arrowangle)},{\ty+sin(\arrowangle)}) node[pos=1, above, font=\scriptsize, inner sep=1pt] {$\alpha$};
    \addlegendimage{only marks, mark=*, draw=PaperRed!80!black, fill=PaperRed, fill opacity=0.3}
    \addlegendentry{target + variants}
    \addlegendimage{only marks, mark=*, draw=PaperOrange!80!black, fill=PaperOrange, fill opacity=0.3}
    \addlegendentry{its neighbors}
    \addlegendimage{only marks, mark=*, draw=PaperGreen!80!black, fill=PaperGreen, fill opacity=0.3}
    \addlegendentry{grassroots}
    \addlegendimage{only marks, mark=*, draw=PaperBlue!80!black, fill=PaperBlue, fill opacity=0.3}
    \addlegendentry{other posts}
    \addlegendimage{only marks, mark=o, mark size=2.5pt, black!70, densely dotted}
    \addlegendentry{uniform}
    \addlegendimage{only marks, mark=o, mark size=2.5pt, black!70}
    \addlegendentry{clone-robust}
    \addlegendimage{black!60, dashdotted}
    \addlegendentry{radius $\alpha$}
    
    \nextgroupplot[
      title={(b) Campaign size},
      xmode=log, log basis x=10, xmin=1, xmax=1001,
      xtick={1,10,100,1000}, xticklabels={1,10,100,1000},
      ymin=0, ymax=100,
      xlabel={\# Posts in the campaign}, ylabel={Campaign share (\%)},
      ylabel shift=-3pt,
      legend pos=north west,
    ]
    \addplot[black, dashed, thick] table[x=n, y=unif, col sep=comma] {\tblAttackSize};
    \addlegendentry{uniform}
    \addplot[name path=loA, draw=none, forget plot] table[x=n, y=lo0.25, col sep=comma] {\tblAttackSize};
    \addplot[name path=hiA, draw=none, forget plot] table[x=n, y=hi0.25, col sep=comma] {\tblAttackSize};
    \addplot[PaperRed!40, fill opacity=0.2, forget plot] fill between[of=loA and hiA];
    \addplot[PaperRed!40, thick, mark=*, mark size=1pt] table[x=n, y=m0.25, col sep=comma] {\tblAttackSize};
    \addlegendentry{$\alpha=0.25$}
    \addplot[name path=loB, draw=none, forget plot] table[x=n, y=lo0.5, col sep=comma] {\tblAttackSize};
    \addplot[name path=hiB, draw=none, forget plot] table[x=n, y=hi0.5, col sep=comma] {\tblAttackSize};
    \addplot[PaperRed!75, fill opacity=0.2, forget plot] fill between[of=loB and hiB];
    \addplot[PaperRed!75, thick, mark=*, mark size=1pt] table[x=n, y=m0.5, col sep=comma] {\tblAttackSize};
    \addlegendentry{$\alpha=0.5$}
    \addplot[name path=loC, draw=none, forget plot] table[x=n, y=lo1, col sep=comma] {\tblAttackSize};
    \addplot[name path=hiC, draw=none, forget plot] table[x=n, y=hi1, col sep=comma] {\tblAttackSize};
    \addplot[PaperRed!70!black, fill opacity=0.2, forget plot] fill between[of=loC and hiC];
    \addplot[PaperRed!70!black, thick, mark=*, mark size=1pt] table[x=n, y=m1, col sep=comma] {\tblAttackSize};
    \addlegendentry{$\alpha=1$}
    
    \nextgroupplot[
      title={(c) One more variant},
      xmin=-0.2, xmax=5, ymin=-23, ymax=2,
      xlabel={Distance to the target $/\,\alpha$}, ylabel={Weight change (\%)},
      ylabel shift=-3pt,
      legend pos=south east,
    ]
    \pgfplotstablecopy{\tblMarginalVariant}\to\mv
    \pgfplotstablegetelem{0}{unif}\of{\mv}\let\uchg\pgfplotsretval
    \addplot[black, dashed, thick, domain=-0.2:5, samples=2] {\uchg};
    \addlegendentry{uniform}
    \addplot[
      scatter, only marks, mark size=1.5pt,
      scatter src=explicit symbolic,
      scatter/classes={
        campaign={mark=*, draw=PaperRed!80!black, fill=PaperRed!60},
        far={mark=*, draw=PaperBlue!80!black, fill=PaperBlue!55},
        near={mark=*, draw=PaperOrange!80!black, fill=PaperOrange!60},
        grass={mark=*, draw=PaperGreen!80!black, fill=PaperGreen!60}},
      forget plot,
    ] table[x=dist, y=crw, meta=group, col sep=comma] {\tblMarginalVariant};
    \addlegendimage{only marks, mark=*, mark size=1.5pt, draw=black!60, fill=black!25}
    \addlegendentry{clone-robust}
    \draw[black!60, dashdotted] (axis cs:1,-23) -- (axis cs:1,2) node[pos=0.06, right, font=\scriptsize, inner sep=1pt] {$\alpha$};
    \draw[black!30] (axis cs:-0.2,0) -- (axis cs:5,0);
    
    \end{groupplot}
    \node[anchor=east] at ($(group c1r1.west)+(-0.15,0)$) {\pgfplotslegendfromname{legendA}};
    \end{tikzpicture}
    \Description{Three panels. Panel (a): a scatter plot of about fifty posts in the plane, drawn as translucent discs whose area is proportional to their weight. A red campaign of twenty variants is packed around a target post, surrounded by a dash-dotted circle of radius alpha; orange posts lie within this circle, a green cluster of six grassroots posts lies elsewhere, and blue posts are spread out. Dotted outlines show uniform weights and solid outlines clone-robust weights; the campaign's total share drops from close to half under uniform weighting to about a tenth, and the grassroots cluster's share also shrinks. Panel (b): campaign share in percent against the number of posts in the campaign, from 1 to 1000 on a logarithmic axis. Under uniform weighting, the share rises towards 100 percent; under clone-robust weighting, three red curves for alpha equal to 0.25, 0.5 and 1 saturate at low levels, with shaded 5 to 95 percent bands. Panel (c): relative weight change in percent of each post of panel (a) against its distance to the target divided by alpha, when one more variant is added. Under uniform weighting, all posts lose the same small amount, shown as a dashed horizontal line; under clone-robust weighting, the campaign posts lose most, nearby posts within distance alpha change slightly, and posts beyond alpha are essentially unaffected.}
    \caption{Clone-robust weighting $f_{\nu,w^{\mathrm{CU}}}$, with $\nu$ uniform on $[0,\alpha]$, on synthetic feeds in $\mathbb{R}^2$.
    (a) A feed of $26$ posts after an attacker injects $20$ variants within distance $\rho = 0.05\alpha$ of a target post; marker area is proportional to weight, and uniform weights are aggregated for the campaign and the grassroots. 
    (b) Campaign share in random feeds of $30$ posts in a $5\times5$ square ($\rho = 0.05$): mean and $5$--$95\%$ range over $40$ feeds.
    (c) Weight change of each post in (a) when a fifth variant is added to a campaign of four.}
    \label{fig:illustration}
\end{figure*}

\section{Explainability of Volumetric Weighting Functions}\label{sec:sharing_measures_gr}

Clone-robust weighting functions were introduced in~\cite{berriaud_clone-robust_2025} to replace ad-hoc empirical weighting with a mathematically rigorous alternative.
By explicitly demanding a pseudo-metric between posts, this approach forces the designer to quantify their assumptions about the underlying space and ensures the resulting weights are consistent with them.
However, a principled weighting procedure does not by itself guarantee transparency. In a news aggregator designed to combat duplication attacks, transparency must extend beyond the algorithm's code to the interpretability of its outcomes, answering for instance: \emph{``Why is my post downweighted?''} Since clone-robust functions force similar posts to share their weight locally, it is natural to ask \emph{how much weight two distinct elements share}.
To address this, we introduce the \emph{sharing coefficient}, which quantifies how two elements of a set $S$ affect each other's weight, and enables comparisons such as whether $x$ ``shares more weight'' with $y$ than with $z$.

We first observe that the clone-robust weighting functions of~\cite{berriaud_clone-robust_2025} naturally admit such a notion of sharing.
Recall the following weighting function on $(\mathbb{R}^n, d_2)$, with $d_2(x,y) = \big(\sum_{i=1}^n (x_i-y_i)^2 \big)^{1/2}$, defined for a finite subset $S \subseteq \mathbb{R}^n$ and an element $x \in S$ by
\begin{equation*}
    g_r(S)(x)
    := \frac{1}{Z}
      \int_{B_r(x)} \frac{1}{|S\cap B_r(z)|}\,dz,
      \qquad Z:= \operatorname{Vol}\!\Big(\bigcup_{y\in S} B_r(y)\Big),
\end{equation*}
where $B_r(\cdot)$ denotes the Euclidean ball of radius $r>0$, and $\operatorname{Vol}(X)$ is the $n$-dimensional Lebesgue measure of $X \subseteq \mathbb{R}^n$.
As in~\cite{berriaud_clone-robust_2025}, $g_r$ can be read as a voting rule: every point $z$ of the space is a voter who splits one unit of vote equally among the $|S\cap B_r(z)|$ posts within distance $r$, and $g_r(S)(x)$ is the normalized total vote received by $x$. Clones thus compete for the same voters, rather than each gaining its own.
This suggests a \emph{sharing coefficient}, where $\chi^{g_r, S}_{y\to x}$ is the vote that $y$ would gain if every voter ignored $x$: a voter $z\in B_r(x)\cap B_r(y)$ seeing $k=|S\cap B_r(z)|$ posts would then give $y$ an additional $\frac{1}{k-1}-\frac{1}{k}=\frac{1}{k(k-1)}$.

\begin{definition}[Sharing coefficient of $g_r$]\label{def:shar_coeff_gr}
    For a finite set $S \subseteq \mathbb{R}^n$ and distinct elements $x,y \in S$, the sharing coefficient $\chi^{g_r, S}$ is defined by
    \begin{equation*}\label{equ:sharing_coeff_gr}
    \begin{gathered}
        \chi^{g_r, S}_{y\to x}
        := \frac{1}{Z} \int_{B_r(x) \cap B_r(y)}
            \frac{1}{|S\cap B_r(z)| \cdot (|S\cap B_r(z)| - 1)} \, dz, \mbox{ and} \\
            \chi^{g_r,S}_{x\to x} := g_r(S)(x) - \sum_{y \neq x} \chi^{g_r, S}_{y\to x}.
    \end{gathered}
    \end{equation*}
    We refer to $\sum_{y \neq x} \chi^{g_r, S}_{x\to y}$ as the shared weight of $x$, and to $\chi^{g_r, S}_{x\to x}$ as its private weight.
\end{definition}

\paragraph{Well-definedness and symmetry.}
For every $z \in B_r(x) \cap B_r(y)$, the set $S \cap B_r(z)$ contains both $x$ and $y$, so the integrand is a positive simple function on the integration domain, and $\chi^{g_r,S}$ is well defined.
Moreover, $x$ and $y$ only enter through the domain of integration and do so symmetrically; hence $\chi^{g_r, S}_{x\to y} = \chi^{g_r, S}_{y\to x}.$

\paragraph{Non-negativity and locality.}
Since the integrand is positive, $\chi^{g_r,S}_{x\to y} \geq 0$, with equality if and only if $B_r(x) \cap B_r(y)$ has measure zero, i.e., $d_2(x,y) \geq 2r$. This is exactly the locality parameter of $g_r$: elements at distance at least $2r$ do not affect each other.
Moreover, the private weight can be rewritten as
\begin{equation*}
\chi^{g_r,S}_{x\to x}
    = \frac{1}{Z} \operatorname{Vol}\!\Big(B_r(x) \setminus \! \bigcup_{y \neq x} B_r(y)\Big),
\end{equation*}
since a voter seeing $k\geq 2$ posts redistributes its entire share $\frac{1}{k}$ of $x$ among the $k-1$ others, so only voters who see $x$ alone contribute.
Hence $\chi^{g_r,S}_{x\to x} \geq 0$, with equality if and only if $B_r(x)$ is covered by the other balls up to a null set.

\paragraph{Effect of element removal.}
By construction, removing $x$ from $S$ has two effects on the weights of the remaining elements. First, voters in $B_r(x)\cap B_r(y)$ no longer vote for $x$, so every $y$ gains its shared weight $\chi^{g_r, S}_{y\to x}$, an \emph{additive local increase}. Second, voters who only saw $x$ disappear, which shrinks the normalization factor to $\operatorname{Vol}\!\big(\bigcup_{y\neq x}B_r(y)\big) = Z \big(1-\chi^{g_r,S}_{x\to x}\big)$, a \emph{global multiplicative rescaling} governed by the private weight of $x$.
Together, these effects yield
\begin{equation}\label{equ:gr_removal_effect}
    g_r(S\setminus \!\{x\})(y) = \Big(g_r(S)(y) + \chi^{g_r,S}_{y \to x}\Big)\Big(1 - \chi^{g_r,S}_{x \to x}\Big)^{-1}.
\end{equation}
Notably, since $0\leq \chi^{g_r,S}_{y \to x} \leq g_r(S)(x) <1$ (also when $y=x$), both local transfer and global rescaling increase the weights of the remaining elements.

\paragraph{Sharing monotonicity.}
At first glance, one might expect a clone-robust weighting function to share more weight with closer neighbors than with more distant ones, i.e., $d(x,y) \le d(x,z)$ should imply $\chi^S_{x \to y} \ge \chi^S_{x \to z}$ for an appropriate sharing coefficient~$\chi$.
However, clones complicate this intuition: if the shared weight is additive over perfect clones, the individual values must decrease as clones are added. Consequently, $\chi^{g_r,S}_{x \to y}$ may be smaller than $\chi^{g_r,S}_{x \to z}$ even though $y$ is closer to $x$ than $z$ is, as depicted in Figure~\ref{fig:illustration_chi_dominance}.
A stronger notion of domination than distance comparison is therefore required, and $\chi^{g_r}$ satisfies the following monotonicity, based on inclusion of intersecting regions.

\begin{figure}[!tbh]
    \centering
    \begin{tikzpicture}[scale=1.2]

        \definecolor{PaperBlue}{HTML}{1f77b4}
        \definecolor{PaperOrange}{HTML}{ff7f0e}
        \definecolor{PaperGreen}{HTML}{2ca02c}
        \definecolor{PaperRed}{HTML}{d62728}

        \coordinate (X) at (0,0);        
        \coordinate (Y_1) at (-1.6, 0.05); 
        \coordinate (Y_2) at (-1.6, -0.05);   
        \coordinate (Z) at (1.8,0);  
        \fill[PaperBlue] (X) circle (1pt) node[above] {$x$};
        \fill[PaperGreen] (Y_1) circle (1pt) node[above left] {$y_1$};
        \fill[PaperOrange] (Y_2) circle (1pt) node[below left] {$y_2$};
        \fill[PaperRed] (Z) circle (1pt) node[above right] {$z$};

        \draw[ PaperBlue] (X) circle (1.5cm);
        \fill[PaperBlue, opacity=0.04] (X) circle (1.5cm);
        \draw[ PaperGreen] (Y_1) circle (1.5cm);
        \fill[PaperGreen, opacity=0.04] (Y_1) circle (1.5cm);
        \draw[ PaperOrange] (Y_2) circle (1.5cm);
        \fill[PaperOrange, opacity=0.04] (Y_2) circle (1.5cm);
        \draw[ PaperRed] (Z) circle (1.5cm);
        \fill[PaperRed, opacity=0.04] (Z) circle (1.5cm);
        
        \node at (-0.8,0.3) { $\chi^{g_r,S}_{\textcolor{PaperBlue}{x} \to \textcolor{PaperGreen}{y_1}} $}; 
        \node at (-0.8,-0.2) { $\propto \frac{1}{3}\cdot\frac{1}{2}$}; 
        \node at (0.9,0.3) { $\chi^{g_r,S}_{\textcolor{PaperBlue}{x} \to \textcolor{PaperRed}{z}} $}; 
        \node at (0.9,-0.2) { $\propto \frac{1}{2}\cdot\frac{1}{1}$};

    \end{tikzpicture}
    \Description{Four overlapping circles of equal radius r. A blue circle is centered at x; a green and an orange circle are centered at two almost identical points y1 and y2 to the left of x; a red circle is centered at z to the right of x, slightly farther from x than y1. The overlap of x's circle with y1's circle is larger than its overlap with z's circle. Annotations indicate that the sharing coefficient between x and y1 is proportional to one third times one half, because three posts share that region, whereas the sharing coefficient between x and z is proportional to one half times one, so z shares more with x than y1 does.}
    \caption{Consider the set $S =\{x,y_1,y_2,z\}.$ Since $y_1$ is closer to $x$ than $z$, it shares more volume with $x$ than $z$ does, i.e., $\operatorname{Vol}\big( B_r(x)\cap B_r(y_1) \big) \geq \operatorname{Vol}\big( B_r(x)\cap B_r(z) \big)$. However, the clone $y_2$ dilutes $y_1$'s sharing coefficient, and $\chi^{g_r,S}_{x\to y_1} \leq \chi^{g_r,S}_{x \to z}.$}
    \label{fig:illustration_chi_dominance}
\end{figure}

\begin{property}[Sharing Domination]\label{lem:sharing_domination_gr}
    Let $x,y,z$ be three distinct elements of a finite set $S \subseteq \mathbb{R}^n$ such that $y$ \textbf{intersection-dominates} $z$ with respect to $x$, i.e.,
    $B_r(x) \cap B_r(z) \subseteq B_r(x) \cap B_r(y).$
    Then $z$ shares less weight with $x$ than $y$ does, that is, $\chi^{g_r,S}_{z \to x} \leq \chi^{g_r,S}_{y\to x}.$
\end{property}

\begin{proof}
The integrands of $\chi^{g_r,S}_{y\to x}$ and $\chi^{g_r,S}_{z \to x}$ are identical and nonnegative, so the claim follows from the inclusion of their integration domains.
\end{proof}

\noindent Intersection-domination is indeed stronger than 
distance-domination: if $y$ intersection-dominates $z$ with respect to $x$ and $B_r(x) \cap B_r(z) \neq \emptyset$, then $d_2(x,y) \leq d_2(x,z)$, since $\operatorname{Vol}\big(B_r(x) \cap B_r(u)\big)$ is strictly decreasing in $d_2(x,u)$ on $[0, 2r)$ and vanishes beyond.
We next show that $\chi^{g_r}$ satisfies another natural form of monotonicity called \emph{dilution}, i.e., sharing decreases as clones are added.

\begin{property}[Sharing Dilution]\label{lem:sharing_dilution_gr}
Let $(M,d)$ be a pseudo-metric space that induces $(\mathbb{R}^n,d_2)$ by metric identification.
Let $x_0$ be in a finite set $S\subseteq M$, and let $(x_i)_{i\in \mathbb{N}_{>0}} \in \big(M \setminus S\big)^{\mathbb{N}_{>0}}$ be a sequence of perfect clones of $x_0$, i.e., $d(x_0,x_i) = 0$ for all $i>0$. We define $S_i = S\cup\{x_1,\dots, x_i\}$ for all $i\in \mathbb{N}.$
\noindent Then, for any $y \in S\setminus\{x_0\}$ and integers $i<j$, it holds that $\chi^{g_r,S_i}_{y \to x_0} \geq \chi^{g_r,S_j}_{y\to x_0}.$
\end{property}

\begin{proof}
Only the integrand of $\chi^{g_r,S_i}_{y \to x_0}$ depends on $i$, and it decreases in $i$. Moreover, $\chi^{g_r,S_0}_{x_0 \to x_0} \geq \chi^{g_r,S_i}_{x_0 \to x_0} = 0$ for all $i>0.$
\end{proof}

We similarly construct in Appendix~\ref{app:shar_coeff_f_nu} a sharing coefficient for the second family of clone-robust weighting functions of~\cite{berriaud_clone-robust_2025}, and show that it inherits most of the desirable properties of $\chi^{g_r}.$

\section{Co-Design of Sharing Measures and Graph Weighting Functions}\label{sec:graph_weighting_design}

Since the requirements of Theorem~\ref{thm:general_weight_function} are easily satisfied (cf. Example~\ref{exple:w^CU}), explainability could actively guide the design of further graph-weighting functions. As such, a natural next step is to define analogous sharing coefficients for the graph setting with a set of natural interpretability axioms. This explicit co-design between weighting functions and their associated sharing measures---enforcing complementary and desirable properties---will increase the overall transparency of the weighting process.

\subsection{Impossibility of Sharing under Arbitrary Perturbations}\label{sec:bad_defs_sharing}

Given the construction in Section~\ref{sec:sharing_measures_gr}, a natural first approach to defining a meaningful sharing coefficient is to quantify the transfer of weight resulting from the deletion of a vertex.
Following the decomposition in Equation~\eqref{equ:gr_removal_effect}, we postulate that such a deletion results in an additive local weight transfer followed by a global multiplicative rescaling.

\begin{definition}[Vertex-Sharing]\label{def:vertex_sharing}
    Let $w$ be a $k$-clone-robust graph weighting function, $G \in \mathcal{G}$, and $x \in V(G)$. We call vertex-sharing coefficients of $x$ any family $\big(\overset{\mathrm{v}}{\chi}{}^{w,G}_{y \to x}\big)_{y \in V(G)}$ such that
    \begin{equation}\label{equ:def_vertex_sharing}
        \overset{\mathrm{v}}{\chi}{}^{w,G}_{y \to x} = \big(1- \overset{\mathrm{v}}{\chi}{}^{w,G}_{x \to x}\big) \cdot w(G\setminus \{x\})(y) - w(G)(y), \quad \forall y \neq x
    \end{equation}
    and $\overset{\mathrm{v}}{\chi}{}^{w,G}_{y \to x} = 0$ whenever $d_G(x,y) > k$.
    We say that $w$ satisfies \textbf{vertex-sharing non-negativity} if, for every graph $G$ and vertex $x \in V(G)$, there exist non-negative vertex-sharing coefficients of $x$.
\end{definition}

\noindent Equation~\eqref{equ:def_vertex_sharing} alone does not determine the coefficients ($\vert V(G)\vert$ unknowns but only $\vert V(G)\vert - 1$ equations per vertex $x$), which motivates the additional locality condition: consistently with the $k$-locality of $w$, vertices far from $x$ should not share any weight with it.
A single such vertex then fixes the private weight $\overset{\mathrm{v}}{\chi}{}^{w,G}_{x \to x}$, while any further one imposes a consistency constraint on $w$: all vertices outside $x$'s $k$-hop neighborhood must scale their weights by a common multiplicative factor upon deletion of $x$.
Moreover, non-negativity can be thought of as a form of monotonicity, requiring that $x$ shares at least as much with vertices inside its $k$-hop neighborhood as with vertices outside of it.
In particular, it implies $w(G)(y) \leq \big(1- \overset{\mathrm{v}}{\chi}{}^{w,G}_{x \to x}\big) \cdot w(G\setminus \{x\})(y) \leq w(G\setminus \{x\})(y)$ for all $y \neq x$, thus the following result shows that this is incompatible with clone-robustness.

\begin{theorem}\label{thm:impossible_vertex_sharing}
    For every $k > 0$ and every $k$-clone-robust graph weighting function $w$, there exist a finite graph $G\in \mathcal{G}$ and distinct vertices $x,y \in V(G)$ such that $w(G \setminus \{x\})(y) < w(G)(y)$.
    In particular, $w$ does not satisfy vertex-sharing non-negativity.
\end{theorem}

\noindent Given this negative result, we seek more permissive notions of sharing. Unlike metric spaces, graphs offer a distinct structural operation: one can also measure the weight transfer resulting from the \emph{deletion of an edge} $e$, where $G - e \coloneq (V(G), E(G)\setminus\{e\})$ and $d_G(e,z) \coloneq \min_{x \in e} d_G(x,z)$.

\begin{definition}[Edge-Sharing]\label{def:edge_sharing}
    Let $w$ be a $k$-clone-robust graph weighting function, $G \in \mathcal{G}$, and $e = \{x,y\} \in E(G)$. We call edge-sharing coefficients of $e$ any family $\big(\overset{\mathrm{e}}{\chi}{}^{w,G}_{z \to e}\big)_{z \in V(G)}$ for which there exists a contracting factor $\eta_{G,e} \in (0,1]$ such that, 
    \begin{equation}\label{equ:def_edge_sharing}
        \overset{\mathrm{e}}{\chi}{}^{w,G}_{z \to e} = w(G-e)(z) - \eta_{G,e}\, w(G)(z), \quad \forall z \in V(G)
    \end{equation}
    and $\overset{\mathrm{e}}{\chi}{}^{w,G}_{z \to e} = 0$ whenever $d_G(e,z) > k$.    
    We say that $w$ satisfies \textbf{bystander contraction} if, for every graph $G$ and edge $e$, there exist edge-sharing coefficients of $e$ such that $\overset{\mathrm{e}}{\chi}{}^{w,G}_{z \to e} \leq 0$ for every $z \notin e$, and that $w$ satisfies \textbf{endpoint expansion} if, for every graph $G$ and edge $e = \{x,y\}$, there exists an expanding factor $\mu_{G,e} \geq 1$ such that $w(G-e)(u) = \mu_{G,e}\, w(G)(u)$ for both $u \in e$.

\end{definition}

\noindent In words, deleting $e$ uniformly rescales the weights of vertices far from $e$ by $\eta_{G,e}$, the surrounding vertices yield relative weight to accommodate this local expansion, and both endpoints gain weight in the same proportion as they become more independent; in particular, endpoint expansion implies $\overset{\mathrm{e}}{\chi}{}^{w,G}_{x \to e} = (\mu_{G,e} - \eta_{G,e})\, w(G)(x) \geq 0$.
Unfortunately, these requirements are again incompatible with clone-robustness.

\begin{theorem}\label{thm:impossible_edge_sharing}
No $1$-clone-robust graph weighting function satisfies both endpoint expansion and bystander contraction.
\end{theorem}

\noindent Both proofs follow a sequence of small graphs, each obtained from the previous one by adding or removing a vertex or an edge, whose weights are pinned down step by step by symmetry, locality and clone-robustness until a contradiction arises; details are deferred to Appendix~\ref{app:proof_main}.

\subsection{Clone-Sharing Coefficients and Desirable Properties}

These negative results, particularly Theorem~\ref{thm:impossible_vertex_sharing}, underscore the difficulty of reconciling locality constraints with weight monotonicity under arbitrary graph perturbations.
A natural solution is to restrict our focus to operations that are inherently compatible with locality, such as the deletion or addition of clones. Because generic nodes may not have existing clones to delete, we opt for the latter operation and use the resulting change in weight as a proxy for the weight already shared with the original node.

\begin{definition}[Clone-Sharing Coefficient]\label{def:clone_sharing}
Let $w$ be a graph weighting function, $G\in \mathcal{G}$ be a finite graph, and $x,y$ be two vertices in $V(G).$
For $i\in \mathbb{N}$, let $G_x^{i}$ denote the graph $G$ augmented with $i$ clones of $x$, i.e., $V(G_x^i)= V(G)\cup\{x^1,\dots,x^i\}$ with $x^j\notin V(G)$ and $x^j \equiv_{G_x^i} x$ for all $1\leq j\leq i.$
We then define the clone-sharing coefficient $\overset{\mathrm{c}}{\chi}{}^{w,G}$ by
$$ \overset{\mathrm{c}}{\chi}{}^{w,G}_{y \to x} := w(G)(y) - w(G_x^1)(y).$$
\end{definition}

\noindent This formulation is consistent with the vertex-deletion paradigm of Equation~\eqref{equ:def_vertex_sharing}: since $G = G_x^1 \setminus\{x^1\}$, the absence of a multiplicative factor corresponds to the requirement $\overset{\mathrm{v}}{\chi}{}^{w,G_x^1}_{x^1 \to x^1} = 0$, which results from the locality constraints of both $w$ and $\overset{\mathrm{v}}{\chi}{}^{w}.$
Moreover, if $w$ is $k$-local, then $\overset{\mathrm{c}}{\chi}{}^{w,G}_{y \to x} =0$ for all $x,y\in V(G)$ such that $d_G(x,y)>k.$

Guided by the interpretability properties established for $\chi^{g_r}$ in Section~\ref{sec:sharing_measures_gr}, we next define similar properties for $\overset{\mathrm{c}}{\chi}{}^{w}$, as requirements on the weighting function $w.$

\begin{definition}\label{prop:non_neg_clone_shar}
A graph weighting function $w$ satisfies \textbf{clone-sharing non-negativity} if $\overset{\mathrm{c}}{\chi}{}^{w,G}_{y \to x}\geq 0$ holds for every finite graph $G\in\mathcal{G}$ and vertices $x,y\in V(G).$
We moreover say that $w$ satisfies $\bm{k}$-\textbf{clone-sharing positivity} if the inequality is strict whenever $d_G(x,y)\leq k.$
\end{definition}

\noindent We next adapt intersection-domination (Property~\ref{lem:sharing_domination_gr}) to graphs, replacing balls by $k$-hop neighborhoods $N_G^k[v] :=\{u \in V(G) \mid d_G(u, v) \leq k\}$, and require the sharing coefficient to be monotone with respect to it.

\begin{definition}\label{prop:clone-sharing_domination}
A graph weighting function $w$ satisfies \textbf{$k$-clone-sharing domination} if $ \overset{\mathrm{c}}{\chi}{}^{w,G}_{z \to x}\leq \overset{\mathrm{c}}{\chi}{}^{w,G}_{y \to x}$ holds for every graph $G\in\mathcal{G}$ and vertices $x,y,z\in V(G)$ such that $N_G^k[x] \cap N_G^k[z] \subseteq N_G^k[x] \cap N_G^k[y]. $
We moreover say that $w$ satisfies \textbf{$k$-strong clone-sharing domination} when equality holds if and only if the two intersections are equal.
\end{definition}

\noindent Parallel to Property~\ref{lem:sharing_dilution_gr}, we finally require that the shared weight decreases monotonically as additional clones are introduced.

\begin{definition}\label{prop:clone-sharing_dilution}
    A graph weighting function $w$ satisfies \textbf{clone-sharing dilution} if $\overset{\mathrm{c}}{\chi}{}^{w,G_x^i}_{y \to x}\geq \overset{\mathrm{c}}{\chi}{}^{w,G_x^j}_{y \to x}$ holds for every graph $G\in\mathcal{G}$, vertices $x,y\in V(G)$ and integers $0\leq i<j.$
    We moreover say that $w$ satisfies $\bm{k}$-\textbf{strong clone-sharing dilution} if the inequality is strict whenever $d_G(x,y)\leq k$.
\end{definition}

\noindent To ground these definitions, we evaluate them against the class-uniform graph weighting function $w^{\mathrm{CU}}$ introduced in Example~\ref{exple:w^CU}.

\begin{example}[Clone-Sharing Properties of $w^{\mathrm{CU}}$]\label{exple:clone_shar_prop_w^CU}
    Adding clones of a vertex $x$ does not change the number of equivalence classes, i.e., $\vert V/{\equiv_{G_x^i}}\vert = \vert V/{\equiv_G}\vert$, but it increases the size of the class of $x$, i.e., $\vert [x]_{G_x^i}\vert = \vert [x]_G\vert + i$.
    Consequently, the weight of any vertex $y \not\equiv_G x$ remains unaffected, and the clone-sharing coefficient takes the simple form
    $$
    \overset{\mathrm{c}}{\chi}{}^{w^{\mathrm{CU}},G}_{y \to x} =
    \begin{cases}
    \frac{1}{\vert V/{\equiv_G}\vert} \left( \frac{1}{\vert [x]_G\vert} - \frac{1}{\vert [x]_G\vert + 1} \right) > 0 & \text{if } y \equiv_G x, \\
    0 & \text{if } y \not\equiv_G x.
    \end{cases}
    $$
    From this expression, we immediately deduce that $w^{\mathrm{CU}}$:
    \begin{itemize}
        \item Satisfies \textbf{clone-sharing non-negativity}, but fails \textbf{$1$-positivity}, since a neighbor $y \not\equiv_G x$ has $\overset{\mathrm{c}}{\chi}{}^{w^{\mathrm{CU}},G}_{y \to x}= 0$.
        \item Satisfies \textbf{clone-sharing dilution}, but fails \textbf{$1$-strong dilution}. For $y \equiv_G x$, the difference $\frac{1}{\vert [x]_G\vert + i} - \frac{1}{\vert [x]_G\vert + i + 1}$ decreases as $i$ grows, but for $y \not\equiv_G x$ the coefficient is zero for every $i$.
        \item Fails \textbf{$1$-clone-sharing domination}. Indeed, let $x$ be a leaf and $y$ its neighbor, with $y\not\equiv_G x$. Taking $z=x$, we have $N_G[x] \cap N_G[z] = N_G[x] = N_G[x] \cap N_G[y]$, so domination requires $0<\overset{\mathrm{c}}{\chi}{}^{w^{\mathrm{CU}},G}_{z \to x} \leq \overset{\mathrm{c}}{\chi}{}^{w^{\mathrm{CU}},G}_{y \to x} =0$, a contradiction.
    \end{itemize}
\end{example}

\noindent We generalize the observations of Example~\ref{exple:clone_shar_prop_w^CU} in three ways. First, clone-sharing dilution is strictly stronger than clone-sharing non-negativity.

\begin{lemma} \label{lem:dilut_imply_pos_shar}
    Clone-sharing dilution (resp. $k$-strong dilution) implies clone-sharing non-negativity (resp. $k$-clone-sharing positivity), but the converse does not hold.
\end{lemma}

\noindent Second, domination is considerably harder to satisfy: the failure of $w^{\mathrm{CU}}$ is not an isolated case, as strong domination at scale $k$ requires locality at scale at least $2k$.

\begin{lemma}\label{lem:incompa_strong_clone_domin}
$k$-strong clone-sharing domination is not compatible with $j$-locality, for any $j \leq 2k-1$.
\end{lemma}

\noindent Finally, Example~\ref{exple:clone_shar_prop_w^CU} shows that clone-sharing dilution does not imply clone-sharing domination; we show that the converse fails too.

\begin{lemma}\label{lem:dilut_indep_domi}
    Clone-sharing dilution and clone-sharing domination do not imply each other.
\end{lemma}

\noindent The proofs of Lemmas~\ref{lem:dilut_imply_pos_shar}, \ref{lem:incompa_strong_clone_domin} and~\ref{lem:dilut_indep_domi} are detailed in Appendix~\ref{app:proof_lem}.

\subsection{Constructing Interpretable Clone-Robust Graph Weighting Functions}

It remains to be seen whether the above clone-sharing properties can actually be satisfied. We next provide a general framework for constructing such graph weighting functions, based on two operators on the space $W$ of graph weighting functions.
First, the \emph{quotient lift operator} $Q$ maps $w \in W$ to the graph weighting function $Qw \in W$ defined, for each graph $G \in \mathcal{G}$ and vertex $x \in V(G)$, by
\begin{equation*}
    Qw(G)(x) \coloneq \frac{w\big(G/{\equiv}\big)\big([x]_{G}\big)}{\big\vert [x]_{G} \big\vert},
\end{equation*}
where $G/{\equiv}$ denotes the quotient graph obtained by collapsing each class of clones $[x]_G$ into a single vertex.
Second, for $k \geq 1$, the \emph{$k$-hop diffusion operator} $D_k$ maps $w \in W$ to the graph weighting function $D_k w \in W$ defined by
\begin{equation*}
    D_k w(G)(x) \coloneq \sum_{y \in N^k_G[x]} \frac{w(G)(y)}{\big\vert N^k_G[y] \big\vert}.
\end{equation*}
That is, each vertex spreads its weight uniformly over its $k$-hop neighborhood, i.e., $D_k$ performs one step of the lazy random walk on the $k$-th power of $G$. 

\begin{theorem}\label{thm:satisfyability_clone_sharing}
    Let $w$ be a symmetric graph weighting function and $k \geq 1$.
    Then $Qw$ is a $1$-clone-robust graph weighting function satisfying clone-sharing dilution. It fails $k$-clone-sharing positivity, and it fails $k$-clone-sharing domination whenever $w$ is positive.
    Moreover, $D_k \circ Qw$ is a $2k$-clone-robust graph weighting function satisfying clone-sharing dilution and $k$-clone-sharing domination. If $w$ is positive, it furthermore satisfies $k$-strong clone-sharing dilution and $k$-strong clone-sharing domination.
\end{theorem}

\noindent In particular, $w^{\mathrm{CU}} = Qw^{\mathrm{U}}$, where $w^{\mathrm{U}}(G)(x) \coloneq 1/\vert V(G) \vert$ denotes the uniform graph weighting function, so Example~\ref{exple:clone_shar_prop_w^CU} is a special case of Theorem~\ref{thm:satisfyability_clone_sharing}.
Moreover, the locality of $D_k \circ Qw$ is optimal, as no $j$-clone-robust graph weighting function with $j \leq 2k-1$ satisfies $k$-strong clone-sharing domination (Lemma~\ref{lem:incompa_strong_clone_domin}).
Note finally that $D_k$ cannot be replaced by the $k$-fold composition $D^k \coloneq D_1 \circ \dots \circ D_1$, even though it diffuses weight over the same $k$-hop neighborhoods.

\begin{lemma}\label{lem:D^k_incompat}
    There exists a symmetric graph weighting function $w$ such that, for every $\ell \geq 2$ and $k \geq 1$, the graph weighting function $D^\ell \circ Qw$ satisfies neither clone-sharing dilution nor $k$-clone-sharing domination.
\end{lemma}

\noindent The proofs of Theorem~\ref{thm:satisfyability_clone_sharing} and Lemma~\ref{lem:D^k_incompat} are detailed in Appendices~\ref{app:proof_main} and~\ref{app:proof_lem} respectively.

\balance

\section{Discussion}\label{sec:discuss}

In this section, we discuss our findings, their practical implementation and open challenges.

\paragraph{Summary of results.}

This work introduces a novel approach to constructing clone-robust weighting functions, extending~\cite{berriaud_clone-robust_2025} along three dimensions: topological invariance (supporting arbitrary metrics), efficient exact evaluation, and Lipschitz---rather than uniform---continuity guarantees.
In the context of CIB, locality bounds the visibility a campaign can gain by mass-producing near-duplicates (Corollary~\ref{cor:attack_budget}), as illustrated on synthetic feeds.
On the explainability side, sharing coefficients for the volumetric construction of~\cite{berriaud_clone-robust_2025} provide intuition for the graph setting, where our impossibility results show that locality is fundamentally at odds with measuring sharing through vertex or edge deletion.
Sharing must thus be defined through the addition of clones, under which dilution and domination can be achieved jointly, at the price of a doubled, yet unavoidable, locality radius.

\paragraph{Practical Instantiation in Newsfeeds.}
Translating this framework into a real-world news aggregator introduces several operational challenges, and our synthetic experiment is merely illustrative.
First, \emph{computational scalability}: while computing weights across millions of posts is prohibitive, the practical scale is arguably bounded by the user's attention budget. Because this is orders of magnitude smaller, our framework can be flexibly integrated downstream, operating directly on a pre-filtered candidate set.
Second, the defense's efficacy fundamentally relies on the \emph{choice of metric}. While standard text similarity may be vulnerable to generative AI obfuscation, 
our metric-agnostic framework can directly build on the similarity measures developed for CIB detection~\cite{mannocci2026detectioncharacterizationcoordinatedonline,luceri_unmasking_2024,cresci_dna-inspired_2016,nizzoli_coordinated_2021};
identifying the most suitable ones, together with the parameters $\alpha$ and $\nu$, warrants an independent empirical investigation.
However, an inherent trade-off remains: any geometric defense will inevitably penalize genuine grassroots coordination if it is metrically indistinguishable from bot activity. Whereas existing algorithms favor faithfully gauging support, our work explores the side of resisting redundancy.
Finally, transparent defenses invite \emph{adversarial adaptation}: attackers may inject clones just beyond the threshold $\alpha$ to evade weight redistribution. If semantic impact decays with distance, however, defenders can fix the maximal spread $\rho$ of clones they wish to protect against, beyond which clones are no longer considered persuasive, and choose $\alpha$ to minimize the bound of Corollary~\ref{cor:attack_budget}.

\paragraph{Future Work.}
This work opens several avenues for future research.
First, our computational results in Appendix~\ref{app:feasability} suggest that Theorem~\ref{thm:impossible_edge_sharing} extends to the weaker requirement of non-negative endpoint sharing, with a contradiction arising on graphs with seven vertices; extracting a human-checkable proof, and extending these impossibility results to $k \geq 2$, remain open.
Second, our clone-sharing coefficients are defined at the level of graph weighting functions; lifting them to the integrated weighting $f_{\nu,w}$, similarly to Appendix~\ref{app:shar_coeff_f_nu}, would provide explanations directly in the metric space of posts.
Third, while the operator $D_k \circ Q$ yields several desirable properties, it still lacks an axiomatic characterization; alternative constructions, such as maximal-clique-based weighting functions (Appendix~\ref{sec:MCC_weighting}), could provide deeper insights into properties like dilution and domination and open the way to such a characterization.
Finally, an empirical study on real data should implement and validate the full defense pipeline, from the choice of metric and of the admissible spread $\rho$ to the calibration of $\alpha$, measuring both the suppression of coordinated campaigns and the cost to genuine content. It should also assess how our theoretical properties translate in practice.

\section*{CRediT Author Statement}

\textbf{Damien Berriaud:} Conceptualization, Methodology, Validation, Formal Analysis, Investigation, Writing, Visualization, Supervision.
\textbf{Supanat Wangsutthitham:} Methodology, Software, Formal Analysis, Investigation, Writing---First Draft, Visualization.
\textbf{Roger Wattenhofer:} Supervision, Funding Acquisition.

\begin{acks}
We thank Felix Quernheim for suggesting the simplified proof of Theorem~\ref{thm:satisfyability_clone_sharing} presented here, which replaces our original, more involved argument.
\end{acks}

\bibliographystyle{ACM-Reference-Format}
\bibliography{cleaned_biblio}


\appendix

\section{Proof of the Main Theorems}\label{app:proof_main}

This appendix contains the proofs of the results of Sections~\ref{sec:general_weighting function} and~\ref{sec:graph_weighting_design}.
We first establish a lemma on cluster graphs, used repeatedly below: on disjoint unions of cliques, symmetry and locality alone determine the weights.
We then prove Theorem~\ref{thm:general_weight_function} and make its evaluation procedure explicit, before turning to the two impossibility results of Section~\ref{sec:bad_defs_sharing} and to the construction of Theorem~\ref{thm:satisfyability_clone_sharing}.

\begin{lemma}\label{lem:cluster_graphs}
Let $w$ be a $k$-clone-robust graph weighting function for some $k > 0$, and let $G$ be a cluster graph with at least one vertex, i.e., a finite graph whose connected components $C_1, \dots, C_m$, with $m \geq 1$, are cliques.
Then $w$ coincides with $w^{\mathrm{CU}}$ on $G$, i.e., for every $i \in \{1,\dots,m\}$ and $x \in C_i$,
$$ w(G)(x) \;=\; \frac{1}{m \cdot \vert C_i \vert} \;=\; w^{\mathrm{CU}}(G)(x). $$
\end{lemma}

\begin{proof}
In a cluster graph, $N_G[x] = C_i$ for every $x \in C_i$. The equivalence classes of $\equiv_G$ are therefore exactly the cliques $C_1,\dots,C_m$, and the second equality holds by definition of $w^{\mathrm{CU}}$.
We prove the first equality by induction on $n(G) := \vert V(G)\vert - m \geq 0$, the number of vertices in excess of one per component.

\emph{Base case.} If $n(G) = 0$, every component is a single vertex and $G$ has no edges. Every permutation of $V(G)$ is then a graph automorphism of $G$, so symmetry of $w$ implies that all $m$ vertices receive the same weight, i.e., $w(G)(x) = 1/m$ for all $x \in V(G)$.

\emph{Induction step.} Let $n(G) \geq 1$, and assume the claim holds for all cluster graphs $G'$ with $n(G') = n(G) - 1$. Some clique $C_i$ then contains two distinct vertices $x \neq y$, which are equivalent in $G$. The graph $G' = G \setminus \{y\}$ is again a cluster graph with $m$ components, $C_i \setminus \{y\}$ and $C_j$ for $j \neq i$, and $n(G') = n(G) - 1$.
Every $z \notin C_i$ lies in another component and satisfies $d_G(x,z) = \infty > k$. By $k$-locality and the induction hypothesis, $w(G)(z) = w(G')(z) = 1/(m \vert C_j\vert)$ for every $z \in C_j$ with $j \neq i$.
Since $w(G)$ is a probability distribution, $C_i$ receives total weight
$$ 1 - \sum_{j\neq i} \vert C_j\vert \cdot \frac{1}{m \vert C_j\vert} \;=\; 1 - \frac{m-1}{m} \;=\; \frac{1}{m}. $$
Finally, every permutation of $C_i$ that fixes all other vertices is an automorphism of $G$, so symmetry splits this weight equally among the vertices of $C_i$, which yields $w(G)(x) = 1/(m\vert C_i\vert)$ for all $x \in C_i$.
\end{proof}

We now turn to Theorem~\ref{thm:general_weight_function}, whose proof treats each property in turn; positivity relies on Lemma~\ref{lem:cluster_graphs}, since every neighborhood graph at a radius below the smallest positive distance in $S$ is a cluster graph. We first recall its statement.

\begin{reptheorem}{thm:general_weight_function}
    Let $w$ be a $k$-clone-robust graph weighting function, $\alpha$ be a positive radius, and $\nu$ be a probability density function with support $[0,\alpha]$ and value bounded by $\overline{\nu}.$

    \noindent For every finite subset $S\subseteq M$ and every $x\in S$, define
    $$       f_{\nu,w}(S)(x) \;=\; \int_0^\alpha \nu(r)\, w(G_r(S))(x)\,dr. $$
    Then $f_{\nu,w}$ is a weighting function of $(M,d)$ and satisfies the following properties for all $S \in\mathcal{P}(M)$ and all $x\in S$:

    \begin{itemize}
        \item \textbf{Positivity}, i.e., $f_{\nu,w}(S)(x) >0$;

        \item \textbf{Symmetry}, i.e.,  $f_{\nu,w} (S)(x) = f_{\nu,w} (S)(\sigma(x))$ for every self-isometry $\sigma:S\mapsto S.$
        
        \item \textbf{Lipschitz Clone Fairness}, i.e., $ \vert f_{\nu,w} (S)(x) - f_{\nu,w}(S)(y) \vert \leq 2 ~\overline{\nu} ~\vert S \vert ~d(x,y)$ for each $y \in S.$        

        \item \textbf{Lipschitz Continuity}, i.e., $\vert f_{\nu,w} (S)(x) - f_{\nu,w} (\pi(S))(\pi(x)) \vert \leq 2 ~\overline{\nu} ~\vert S \vert^2 ~\max_{x\in S} d(x,\pi(x))$,
        for each bijection $\pi:S\mapsto \pi(S) \subseteq M.$ 
        
        \item \textbf{Lipschitz $(k\alpha)$-Locality}, i.e.,
        if  $d(x,y) \geq k \alpha $, then $\vert f_{\nu,w} (S \cup Z)(y) - f_{\nu,w}(S)(y) \vert \leq$ \\
        $2 \overline{\nu} \vert S \vert \max_{z \in Z} d(x, z)$ for each $y \in S$

    \end{itemize}

    \noindent Moreover, $f_{\nu,w} (S)(x)$  can be computed in time $O\big( m^2 (W(m) + C_\nu + C_d + \log m )\big)$,  where $W(m)$, $C_\nu$ and  $C_d$ are the respective worst-case evaluation complexity of: $w$ over all graphs with $m$ vertices; the cumulative distribution of $\nu$ at a single real argument; and the pseudo-metric $d.$
\end{reptheorem}

\begin{proof}

Let $S$ be a finite subset of $M$ and $x$ be an element in $S.$ 
We denote by $\mathcal{D}^x_S = \{d(x,z) \mid z \in S\}$ the set of distances between $x$ and the elements of $S$, and by $\mathcal{D}_S = \bigcup_{x\in S}\mathcal{D}^x_S$ the set of pairwise distances in $S$. Note that $0 = d(x,x)$ belongs to both sets, and that $\vert \mathcal{D}^x_S\setminus\{0\}\vert \leq \vert S\vert - 1$ and $\vert \mathcal{D}_S\setminus\{0\}\vert \leq \vert S\vert(\vert S\vert-1)/2$  by \emph{identity} and \emph{symmetry} of the pseudo-metric $d.$

\begin{itemize}

    \item \textbf{Well-definedness and Positivity:}
    Let $0 < r_1 < \dots < r_\ell$ be the increasing list of the distinct positive values in $\mathcal{D}_S$, where $\ell := \vert \mathcal{D}_S \setminus \{0\}\vert \leq \vert S\vert(\vert S\vert-1)/2$, and set $r_0 := 0$ and $r_{\ell+1} := \infty$.
    The edge set $E_r(S)$ is constant on each semi-open interval $I_j = [r_j, r_{j+1})$, $j \in \{0,\dots,\ell\}$, which partition $\mathbb{R}_{\geq 0}$. Hence the map $r \mapsto w(G_r(S))(x)$ is piecewise constant with at most $\ell+1$ pieces.
    Since $\nu$ is a bounded measurable function, the map $r \mapsto \nu(r)\, w(G_r(S))(x)$ is integrable and $f_{\nu,w}(S)(x)$ is well defined.
    
    For positivity, note that for $r \in I_0 = [0, r_1)$, two elements of $S$ are adjacent in $G_r(S)$ if and only if they are at distance zero. By the triangle inequality, being at distance zero is an equivalence relation, so $G_r(S)$ is a cluster graph, and Lemma~\ref{lem:cluster_graphs} yields $w(G_r(S))(x) \geq 1/\vert S\vert^2$.
    Since $\nu$ and $w$ are non-negative and the support of $\nu$ is $[0,\alpha]$, the interval $[0, \min(r_1,\alpha))$ has positive $\nu$-measure, and we conclude that
    \begin{equation*}
        \begin{aligned}
            f_{\nu,w}(S)(x) \;&\geq\; \int_0^{\min(r_1,\alpha)} \nu(r)\, w(G_r(S))(x)\, dr \\
            \;&\geq\; \frac{1}{\vert S\vert^2}\int_0^{\min(r_1,\alpha)} \nu(r)\, dr \;>\; 0.
        \end{aligned}
    \end{equation*} 
    
    \item \textbf{Evaluation Complexity:}
    To evaluate $f_{\nu,w}(S)(x)$, we proceed in three steps, as detailed in Algorithm~\ref{alg:evaluate_f_nu_w}.
    Let $\Gamma$ denote the cumulative distribution function of $\nu$.
    
    \begin{enumerate}
        \item \emph{Distances.}
        Compute $d(u,v)$ for all $u,v \in S$, which yields $\mathcal{D}_S$ in time $O(\vert S\vert^2 C_d)$.
        Discard the values greater than or equal to $\alpha$ and sort the remaining ones, at an additional cost $O(\vert S\vert^2 \log \vert S\vert)$. This yields the threshold radii $0 = r_0 < r_1 < \dots < r_{\ell_\alpha} < \alpha$, where $\ell_\alpha \leq \ell \leq \vert S\vert(\vert S\vert-1)/2$; we further set $r_{\ell_\alpha+1} := \alpha$.
    
        \item \emph{Graphs at threshold radii.}
        For each $i \in \{0,\dots,\ell_\alpha\}$, construct $G_{r_i}(S)$ incrementally from $G_{r_{i-1}}(S)$ by adding the edges between elements at distance exactly $r_i$.
        Then evaluate $w(G_{r_i}(S))(x)$ and $\Gamma(r_{i+1})$, at cost $O(W(\vert S\vert) + C_\nu)$ per index.
    
        \item \emph{Integration.}
        Since $G_r(S) = G_{r_i}(S)$ for all $r \in [r_i, r_{i+1})$, we have
        $f_{\nu,w}(S)(x) = \sum_{i=0}^{\ell_\alpha} \big(\Gamma(r_{i+1})-\Gamma(r_i)\big)\, w(G_{r_i}(S))(x)$, which is accumulated along the loop.
    \end{enumerate}
    
    Overall, the above procedure results in worst-case time complexity $O\big( \vert S\vert^2 (W(\vert S\vert) + C_\nu + C_d + \log \vert S\vert )\big).$

    \item \textbf{Symmetry:} Let $\sigma:S \mapsto S$ be a self-isometry and $r\geq0$ be a positive radius.
    Note that $\sigma$ is a graph automorphism of $G_r(S)$: indeed, $u,v \in S$ are adjacent in $G_r(S)$ if and only if $ d(\sigma(u),\sigma(v)) = d(u,v) \leq r$, that is if and only if $\sigma(u)$ and $ \sigma(v) $ are adjacent in $G_r(S).$ 
    By symmetry of $w$, we then have, i.e.,
    \begin{equation*}
        \begin{aligned}
            f_{\nu,w}(S)(\sigma(x)) 
            &= \int_0^\alpha \nu(r)~ w(G_r(S))(\sigma(x)) dr, \\
            &= \int_0^\alpha \nu(r)~ w(G_r(S))(x) dr =  f_{\nu,w}(S)(x),
        \end{aligned}
    \end{equation*}
    and $f_{\nu,w}(S)$ is symmetric.

    \item \textbf{Lipschitz Clone-Fairness:} 
    Let $y$ be an element of $S$, and let $\overline{B}_{d(x,y)}(\mathcal{D}^x_S) = \bigcup_{z\in S} [ d(x,z) - d(x,y), d(x,z) + d(x,y)] $ denote the union of closed balls of $\mathbb{R}$ of radius $d(x,y)$ centered around the elements of $\mathcal{D}^x_S$.
      
    For a radius $r$ in $[0,\alpha] \setminus \! \overline{B}_{d(x,y)}(\mathcal{D}^x_S)$, note that $x$ and $y$ are equivalent in $G_r(S).$ Indeed, for $u \in S$ with $r\geq d(x,u)$, we also have $r\geq d(x,u) + d(x,y) \geq d(y,u)$ by definition of $r$ and the triangle inequality, hence $u$ is connected to both $x$ and $y$ in $G_r(S).$ For $u \in S$ such that $r< d(x,u)$, a similar argument gives $r< d(x,u) - d(x,y) \leq d(y,u)$ and $u$ is a neighbor of neither $x$ nor $y.$ In particular, the case $u = x$, with $d(x,x) = 0 \in \mathcal{D}^x_S$, ensures that $x$ and $y$ are adjacent.
      
    Let $\sigma:S\to S$ be the transposition that exchanges $x$ and $y$, i.e., $\sigma(x) = y$, $\sigma(y)= x$ and $\sigma(z) = z$ for all $z\in S \setminus \!\{x,y\}.$ 
      Since $x$ and $y$ have the same closed neighborhood, $\sigma$ is a graph automorphism of $G_r(S)$ for all $r\in [0,\alpha] \setminus \! \overline{B}_{d(x,y)}(\mathcal{D}^x_S)$, and the \emph{symmetry} of $w$ implies
     \begin{equation}
        \label{equ:general_clone_fair}
         \begin{aligned}
              &\vert f_{\nu,w} (S)(x) - f_{\nu,w}(S)(y) \vert \\
              &\leq \int_0^\alpha \nu(r)~ \big \vert w(G_r(S))(x) -  w(G_r(S)) (\sigma(x)) \big \vert dr ,\\
              &\stackrel{(a)}{\leq}  \int_{[0,\alpha] \cap \overline{B}_{d(x,y)}(\mathcal{D}^x_S)} ~ \nu(r) ~dr ,\\
              &\stackrel{(b)}{\leq} 2 ~\overline{\nu} ~\vert S \vert ~d(x,y).
         \end{aligned}
     \end{equation}
    Inequality $(a)$ uses that the difference of weights is null on $[0,\alpha] \setminus \! \overline{B}_{d(x,y)}(\mathcal{D}^x_S)$ by \emph{symmetry} of $w$, and bounded by one on $[0,\alpha] \cap \overline{B}_{d(x,y)}(\mathcal{D}^x_S).$ Inequality $(b)$ uses the bound on the probability density function $\overline{\nu}$, as well as the fact that the Lebesgue measure of $[0,\alpha] \cap \overline{B}_{d(x,y)}(\mathcal{D}^x_S)$ is bounded by that of $ \overline{B}_{d(x,y)}(\mathcal{D}^x_S)$, which is in turn bounded by $\vert \mathcal{D}^x_S \vert \leq \vert S \vert$ times the measure of each interval $2 d(x,y) .$

  \item \textbf{Lipschitz Continuity:} 
 Let $ T$ be a finite subset of $M$ of cardinality $\vert T \vert = \vert S \vert$, and let $ \pi : S \to T$ be a bijective map. We denote the maximal distance between elements of $S$ and their image in $T$ as  $\delta := \max_{u\in S} d(u,\pi(u)) $ and consider  $\overline{B}_{2\delta}(\mathcal{D}_S) = \bigcup_{u,v\in S} [ d(u,v) -  2\delta, d(u,v) + 2\delta] $ the union of closed balls of $\mathbb{R}$ of radius $2\delta$ centered around the elements of $\mathcal{D}_S.$

 For a radius $r\in [0,\alpha]\setminus \! \overline{B}_{2\delta}(\mathcal{D}_S)$, note that $G_r(S)$ and $G_r(T)$ are isomorphic, i.e., $u,v \in S$ are connected in $G_r(S)$ if and only if $\pi(u)$ and $\pi(v)$ are connected in $G_r(T).$ Indeed, if $r\geq d(u,v)$, then we also have $r\geq d(u,v) + 2 \delta \geq d(u,v) + d(u,\pi(u)) + d(v,\pi(v)) \geq d(\pi(u),\pi(v))$ by definition of $r$ and $\delta$, as well as by the triangle inequality. Hence $\pi(u)$ is connected to $\pi(v)$ in $G_r(T).$ Similarly if $r< d(u,v)$, we also have $r< d(u,v) -2 \delta \leq d(u,v) - d(u,\pi(u)) - d(v,\pi(v)) \leq d(\pi(u),\pi(v))$ and $\pi(u)$ is not a neighbor of $\pi(v)$ in $G_r(T).$
 
  Similarly to Equation~\ref{equ:general_clone_fair}, we conclude using the \emph{symmetry} of $w$, i.e.,
   \begin{equation*}
     \begin{aligned}
          \big\vert f_{\nu,w} (S)(x) - f_{\nu,w}(\pi(S))(\pi(x)) \big\vert 
          &\leq  \int_{[0,\alpha] \cap \overline{B}_{2\delta}(\mathcal{D}_S)} \nu(r) dr , \\
          &\leq 2 ~\overline{\nu} ~\vert S \vert^2 ~ \max_{u\in S} d(u,\pi(u)).
     \end{aligned}
 \end{equation*}
 The last inequality uses that the Lebesgue measure of $\mathbb{R}_{\geq 0} \cap \overline{B}_{2\delta}(\mathcal{D}_S)$ is bounded by the cardinality $ \vert \mathcal{D}_S \setminus\{0\} \vert \leq \vert S \vert (\vert S \vert -1) /2$ times the diameter of each interval $4 \delta$, plus $2\delta$ for the interval $[0,2\delta]$ centered at $0$, which in total is at most $2\vert S\vert^2 \delta$.

     \item \textbf{Lipschitz $(k\alpha)$-Locality:}
     Let $y$ be an element of $S$ such that $d(x, y) \geq k \alpha$, let $Z = \{z_1, \dots, z_{\vert Z\vert}\} \subseteq M \setminus S$ be finite, and let $z \in Z$ be a maximizer of $d(x, z)$.
     Denote $\Lambda := \overline{B}_{d(x, z)}(\mathcal{D}^x_S) \cup [0, 2d(x,z)]$.
     For a radius $r$ in $[0,\alpha) \setminus \Lambda$ and any $Z' \subseteq Z$, note first that the graph distance between $x$ and $y$ in $G_r(S \cup Z')$ strictly exceeds $k$. This is due to triangle inequality: suppose $d_{G_r(S\cup Z')}(x, y) = l$. Observe that any path of length $l$ in $G_r(S \cup Z')$ can cover the distance in metric space at most $lr$. Since $\alpha > r$, we have $l\alpha > lr$. The path connecting $x$ and $y$ must be of length, by definition, at least $d(x, y) \geq k \alpha$, so we have
     \[
        l\alpha > lr \geq d(x, y) \geq k \alpha.
     \]
     Dividing by $\alpha$ gives $l > k$, so we know that $d_{G_r(S\cup Z')}(x, y) > k$.
     Furthermore, every $z' \in Z'$ is equivalent to $x$ in $G_r(S \cup Z')$. Indeed, for $u \in S$, the triangle inequality gives $\vert d(z',u) - d(x,u)\vert \leq d(x,z') \leq d(x,z)$, while $r \notin [d(x,u) - d(x,z), d(x,u) + d(x,z)]$, so $u$ is adjacent to $z'$ if and only if it is adjacent to $x$; and since $r > 2d(x,z)$, the elements of $Z' \cup \{x\}$ are pairwise adjacent. In particular, $z_i$ is equivalent to $x$ in $G_r(S \cup \{z_1, \dots, z_i\})$ for every $i$, irrespective of the order in which the elements of $Z$ are enumerated.
     Denote $f_i \coloneq f_{\nu, w}(S \cup \{z_1, z_2, \dots, z_i\})(y)$ and $w_i^r \coloneq w(G_r(S \cup \{z_1, z_2, \dots, z_i\}))(y)$. We follow a similar approach to one in Equation~\ref{equ:general_clone_fair}, while also expanding the integrand to a telescopic sum using triangle inequality and iteratively use \emph{$k$-locality} of $w$ to collapse them to 0, i.e.,
     \begin{equation*}
         \begin{aligned}
              \big\vert f_{|Z|} - f_0 \big\vert 
              &\leq \int_0^\alpha \nu(r) \bigl| w_{|Z|}^r - w_0^r \bigr| ~dr \\
              &= \int_{[0, \alpha) \setminus \Lambda} \nu(r) \bigl| w_{|Z|}^r - w_0^r \bigr| ~dr \\
              & \qquad + \int_{[0, \alpha) \cap \Lambda} \nu(r) \underbrace{\bigl| w_{|Z|}^r - w_0^r \bigr|}_{\leq 1} ~dr \\
              &\leq \int_{[0, \alpha) \setminus \Lambda} \nu(r) \sum_{i=1}^{|Z|} \bigl| \underbrace{w_i^r - w_{i-1}^r}_{0} \bigr| ~dr \\
              & \qquad + \int_{[0, \alpha) \cap \Lambda} \nu(r) ~dr \\
              &\leq 0 + 2 ~\overline{\nu} ~\vert S \vert ~d(x,z).
         \end{aligned}
     \end{equation*}
     The last inequality uses that the Lebesgue measure of $\Lambda \cap \mathbb{R}_{\geq 0}$ is at most $\vert \mathcal{D}^x_S \setminus\{0\}\vert \cdot 2d(x,z) + 2d(x,z) \leq 2\vert S\vert\, d(x,z)$.

\end{itemize}
\end{proof}

The complexity bound relies on the fact that $r \mapsto G_r(S)$ is piecewise constant, changing only at the pairwise distances in $S$. Algorithm~\ref{alg:evaluate_f_nu_w} makes the resulting procedure explicit: it sorts these distances once, then evaluates $w$ on each distinct neighborhood graph, weighted by the $\nu$-mass of the corresponding radius interval.

\begin{algorithm}[!h]
\SetAlgoNoLine
\KwIn{Finite subset $S \subseteq M$, element $x \in S$, graph weighting function $w$, radius $\alpha > 0$, cumulative distribution function $\Gamma$ associated with probability density $\nu$ over $[0,\alpha]$.}
\KwOut{Value of $f_{\nu,w}(S)(x)$.}

Initialize empty dictionary $\mathcal{E}: \mathbb{R} \to 2^{S \times S}$\;
\For{each $y, z \in S$}{
    Compute $r \gets d(y, z)$\;
    \If{$r < \alpha$}{
        Append $(y, z)$ to $\mathcal{E}(r)$\;
    }
}
Extract and sort keys of $\mathcal{E}$ in ascending order to obtain distinct radii $0 = r_0 < r_1 < \dots < r_{\ell_\alpha} < r_{\ell_\alpha+1} = \alpha$\;

Initialize $f_{\nu,w}(S)(x) \gets 0$\;
Initialize graph $G$ with vertex set $S$ and edge set $\mathcal{E}(r_0)$\;
$\Gamma_{\text{curr}} \gets 0$\;

\For{$i = 0$ to $\ell_\alpha$}{
    Compute $w_i \gets w(G)(x)$\;
    Compute $\Gamma_{\text{next}} \gets \Gamma(r_{i+1})$\;
    Update $f_{\nu,w}(S)(x) \gets f_{\nu,w}(S)(x) + (\Gamma_{\text{next}} - \Gamma_{\text{curr}}) \cdot w_i$\;
    Update $\Gamma_{\text{curr}} \gets \Gamma_{\text{next}}$\;
    Add edge set $\mathcal{E}(r_{i+1})$ to graph $G$\;
}

\textbf{return} $f_{\nu,w}(S)(x)$\;

\caption{Evaluate $f_{\nu,w}(S)(x)$}
\label{alg:evaluate_f_nu_w}
\end{algorithm}

We next prove the two impossibility results of Section~\ref{sec:bad_defs_sharing}. Both proofs follow the same pattern: starting from cluster graphs, whose weights are fixed by Lemma~\ref{lem:cluster_graphs}, we add or remove single vertices or edges and propagate the resulting constraints until a contradiction arises. We begin with vertex sharing, for which a single graph suffices.

\begin{reptheorem}{thm:impossible_vertex_sharing}
For every $k > 0$ and every $k$-clone-robust graph weighting function $w$, there exists a finite graph $G\in \mathcal{G}$ and distinct vertices $u,v \in V(G)$ such that $w(G \setminus \{u\})(v) < w(G)(v).$
\end{reptheorem}

\begin{proof}
    Let $w$ be $k$-clone-robust, and suppose for contradiction that $w(G \setminus \{u\})(v) \geq w(G)(v)$ for every graph $G$ and distinct vertices $u, v \in V(G)$.
    Consider the graph $G \coloneq P_3 \cup 3K_1$, consisting of a path $a - b - c$ and three isolated vertices $d, e, f$:
    \[
    \begin{tikzpicture}[every node/.style={circle,thick,draw,inner sep=0pt,minimum size=5mm,text height=1.5ex,text depth=0.25ex}]
        \node (a) at (0,0) {$a$}; \node (b) at (1,0) {$b$}; \node (c) at (2,0) {$c$};
        \node (d) at (3.2,0) {$d$}; \node (e) at (4.2,0) {$e$}; \node (f) at (5.2,0) {$f$};
        \draw[thick] (a) -- (b) -- (c);
    \end{tikzpicture}
    \]
    By symmetry of $w$, the isolated vertices $d, e, f$ receive a common weight $\alpha$, the leaves $a, c$ a common weight $\beta$, and we write $\gamma \coloneq w(G)(b)$, so that
    $$3\alpha + 2\beta + \gamma = 1.$$
    
    We now remove single vertices from $G$; in both cases, the resulting graph is a cluster graph, whose weights are given by Lemma~\ref{lem:cluster_graphs}.
    \begin{itemize}
        \item Removing the leaf $c$ yields $G \setminus \{c\} = K_2 \cup 3K_1$, with four components, so that $a$ and $b$ receive weight $\frac{1}{4} \cdot \frac{1}{2} = \frac18$. By assumption, $\beta = w(G)(a) \leq \frac18$ and $\gamma = w(G)(b) \leq \frac18$.
        \item Removing the center $b$ yields $G \setminus \{b\} = 5K_1$, so that $d$ receives weight $\frac15$. By assumption, $\alpha = w(G)(d) \leq \frac15$.
    \end{itemize}
    Combining these bounds, we get
    $$1 = 3\alpha + 2\beta + \gamma \leq \frac35 + \frac28 + \frac18 = \frac{39}{40},$$
    a contradiction.
\end{proof}

\begin{figure*}[htbp]
    \centering
    \begin{tikzpicture}[every node/.style={circle,thick,draw,inner sep=0pt,minimum size=5mm,text height=1.5ex,text depth=0.25ex,font=\scriptsize}]
        \begin{scope}
            \node (1) at (-0.9,0) {1}; \node (2) at (0,0) {2}; \node (3) at (0.9,0) {3};
            \draw[thick] (1)--(2)--(3);
            \node[draw=none] at (0,-1.2) {$P_3$};
        \end{scope}
        \begin{scope}[xshift=4cm]
            \node (1) at (-0.9,0) {1}; \node (2) at (0,0) {2}; \node (3) at (0.9,0) {3}; \node (4) at (-0.45,-0.75) {4};
            \draw[thick] (1)--(2)--(3) (1)--(4)--(2);
            \node[draw=none] at (0,-1.2) {$G_p$};
        \end{scope}
        \begin{scope}[xshift=8cm]
            \node (4) at (-1.35,0) {4}; \node (1) at (-0.45,0) {1}; \node (2) at (0.45,0) {2}; \node (3) at (1.35,0) {3};
            \draw[thick] (4)--(1)--(2)--(3);
            \node[draw=none] at (0,-1.2) {$P_4$};
        \end{scope}
        \begin{scope}[xshift=12cm]
            \node (4) at (-1.35,0) {4}; \node (1) at (-0.45,0) {1}; \node (2) at (0.45,0) {2}; \node (3) at (1.35,0) {3};
            \draw[thick] (4)--(1)--(2);
            \node[draw=none] at (0,-1.2) {$G_0$};
        \end{scope}
        \begin{scope}[yshift=-2.6cm]
            \node (1) at (0,0) {1}; \node (2) at (0.9,0.45) {2}; \node (3) at (0.9,-0.45) {3}; \node (4) at (-0.9,0) {4};
            \draw[thick] (1)--(2) (1)--(3) (1)--(4);
            \node[draw=none] at (0,-1.2) {$S$};
        \end{scope}
        \begin{scope}[xshift=4cm,yshift=-2.6cm]
            \node (1) at (0,0) {1}; \node (2) at (0.9,0.45) {2}; \node (3) at (0.9,-0.45) {3};
            \node (4) at (-0.9,0.45) {4}; \node (5) at (-0.9,-0.45) {5};
            \draw[thick] (1)--(2) (1)--(4) (1)--(5) (4)--(5);
            \node[draw=none] at (0,-1.2) {$G_1$};
        \end{scope}
        \begin{scope}[xshift=8cm,yshift=-2.6cm]
            \node (1) at (0,0) {1}; \node (2) at (0.9,0.45) {2}; \node (3) at (0.9,-0.45) {3};
            \node (4) at (-0.9,0.45) {4}; \node (5) at (-0.9,-0.45) {5};
            \draw[thick] (1)--(2) (1)--(3) (1)--(4) (1)--(5) (4)--(5);
            \node[draw=none] at (0,-1.2) {$G_c$};
        \end{scope}
        \begin{scope}[xshift=12cm,yshift=-2.6cm]
            \node (1) at (0,0) {1}; \node (2) at (0.9,0.45) {2}; \node (3) at (0.9,-0.45) {3};
            \node (4) at (-0.9,0.45) {4}; \node (5) at (-0.9,-0.45) {5};
            \draw[thick] (1)--(2) (1)--(3) (1)--(4) (1)--(5) (4)--(5) (2)--(3);
            \node[draw=none] at (0,-1.2) {$G_b$};
        \end{scope}
    \end{tikzpicture}
    \Description{Eight small graphs arranged in two rows of four. Top row: the path P3 on vertices 1, 2, 3; the graph Gp, which adds a vertex 4 adjacent to vertices 1 and 2; the path P4 with vertices 4, 1, 2, 3 in order; and the graph G0, a path 4, 1, 2 together with the isolated vertex 3. Bottom row: the star S with center 1 and leaves 2, 3 and 4; the graph G1, where vertex 1 is adjacent to 2, 4 and 5, vertices 4 and 5 are adjacent, and vertex 3 is isolated; the graph Gc, which is G1 with the additional edge between 1 and 3; and the bowtie-like graph Gb, which is Gc with the additional edge between 2 and 3.}
    \caption{The sequence of graphs used in the proof of Theorem~\ref{thm:impossible_edge_sharing}.}
    \label{fig:edge_sharing_graphs}
\end{figure*}

The proof of Theorem~\ref{thm:impossible_edge_sharing} applies the same strategy to the longer sequence of graphs depicted in Figure~\ref{fig:edge_sharing_graphs}, where each arrow indicates the operation relating two graphs. Unlike the previous proof, it uses $1$-locality between vertices at distance two, which is why it is only stated for $k = 1$.

\begin{reptheorem}{thm:impossible_edge_sharing}
  There exists no $1$-clone-robust graph weighting function that simultaneously satisfies endpoint expansion and bystander contraction.
\end{reptheorem}

\begin{proof}
Suppose for contradiction that $w$ is a $1$-clone-robust graph weighting function satisfying both endpoint expansion and bystander contraction. We successively determine the weights of the graphs depicted in Figure~\ref{fig:edge_sharing_graphs}. Throughout, $\mu$ denotes the expanding constant of the edge considered in each step, and weights of cluster graphs are given by Lemma~\ref{lem:cluster_graphs}.
\begin{itemize}
    \item \textbf{Path $P_3$.} Removing $e = \{2,3\}$ yields $K_2 \cup K_1$, in which vertices $2$ and $3$ have weights $\frac14$ and $\frac12$. Endpoint expansion gives $\frac14 = \mu\, w(P_3)(2)$ and $\frac12 = \mu\, w(P_3)(3)$, so that $w(P_3)(3) = 2\,w(P_3)(2)$. By symmetry and normalization, $w(P_3)(2) = \frac15$ and $w(P_3)(1) = w(P_3)(3) = \frac25$.

    \item \textbf{Paw $G_p$}, obtained by adding a clone $4$ of vertex $1$ to $P_3$. Since $d_{G_p}(1,3) = 2$, $1$-locality gives $w(G_p)(3) = \frac25$. Removing $\{2,3\}$ yields $K_3 \cup K_1$, so that endpoint expansion gives $\frac12 = \mu \cdot \frac25$ and $\frac16 = \mu\, w(G_p)(2)$, i.e., $\mu = \frac54$ and $w(G_p)(2) = \frac{2}{15}$. By symmetry and normalization, $w(G_p)(1) = w(G_p)(4) = \frac{7}{30}$.

    \item \textbf{Path $P_4 = G_p - \{2,4\}$.} Endpoint expansion gives $w(P_4)(4) = \frac{7}{30}\mu$ and $w(P_4)(2) = \frac{2}{15}\mu$. By symmetry of $P_4$ and normalization, $2\big(\frac{7}{30} + \frac{4}{30}\big)\mu = 1$, so that $\mu = \frac{15}{11}$, $w(P_4)(3) = w(P_4)(4) = \frac{7}{22}$, and $w(P_4)(1) = w(P_4)(2) = \frac{2}{11}$.

    \item \textbf{Graph $G_0 = P_4 - \{2,3\}$}, consisting of the path $4 - 1 - 2$ and the isolated vertex $3$. Endpoint expansion on $\{2,3\}$ gives $w(G_0)(2) / w(G_0)(3) = \frac{2/11}{7/22} = \frac47$. Moreover, $G_0 - \{1,2\} = K_2 \cup 2K_1$ assigns $\frac16$ to vertex $1$ and $\frac13$ to vertex $2$, so that endpoint expansion on $\{1,2\}$ gives $w(G_0)(1) = \frac12\, w(G_0)(2)$. Using $w(G_0)(4) = w(G_0)(2)$ by symmetry, normalization yields $w(G_0)(3) = \frac{7}{17}$, $w(G_0)(2) = w(G_0)(4) = \frac{4}{17}$, and $w(G_0)(1) = \frac{2}{17}$.

    \item \textbf{Star $S = K_{1,3} = G_0 + \{1,3\}$.} Endpoint expansion on $\{1,3\}$ gives $w(S)(3)/w(S)(1) = \frac{7/17}{2/17} = \frac72$. By symmetry and normalization, $w(S)(1) = \frac{2}{23}$ and every leaf of $S$ has weight $\frac{7}{23}$.

    \item \textbf{Graph $G_1$}, obtained by adding a clone $5$ of vertex $4$ to $G_0$. Since $d_{G_1}(4,2) = 2$ and $d_{G_1}(4,3) = \infty$, $1$-locality gives $w(G_1)(2) = \frac{4}{17}$ and $w(G_1)(3) = \frac{7}{17}$. Removing $\{1,2\}$ yields $K_3 \cup 2K_1$, so that endpoint expansion gives $\frac13 = \mu \cdot \frac{4}{17}$ and $\frac19 = \mu\, w(G_1)(1)$, i.e., $\mu = \frac{17}{12}$ and $w(G_1)(1) = \frac{4}{51}$.

    \item \textbf{Cricket $G_c = G_1 + \{1,3\}$}, which is also obtained by adding a clone $5$ of vertex $4$ to $S$. Since $d_{G_c}(4,2) = d_{G_c}(4,3) = 2$, $1$-locality gives $w(G_c)(2) = w(G_c)(3) = \frac{7}{23}$. Endpoint expansion on $\{1,3\}$ gives $\frac{7}{17} = \mu \cdot \frac{7}{23}$ and $\frac{4}{51} = \mu\, w(G_c)(1)$, i.e., $\mu = \frac{23}{17}$ and $w(G_c)(1) = \frac{4}{69}$. By symmetry and normalization, $w(G_c)(4) = w(G_c)(5) = \frac16$.

    \item \textbf{Bowtie $G_b = G_c + \{2,3\}$.} The bowtie is also obtained by adding a clone of the pendant vertex $3$ to the paw $G_p$, whose vertices $1$ and $4$ lie at distance $2$ from it. Hence, $1$-locality and symmetry imply that all four non-central vertices of $G_b$ have weight $\frac{7}{30}$, and normalization gives $w(G_b)(1) = \frac{1}{15}$.
    We finally apply bystander contraction to $e = \{2,3\}$, with $G_b - e = G_c$ and contracting constant $\eta$. Since $4 \notin N_{G_b}[2] \cup N_{G_b}[3]$, equality holds for vertex $4$, i.e., $\frac16 = \eta \cdot \frac{7}{30}$, so that $\eta = \frac57$. For the non-incident vertex $1$, bystander contraction then requires
    \[
        \frac{4}{69} = w(G_c)(1) \leq \eta\, w(G_b)(1) = \frac57 \cdot \frac{1}{15} = \frac{1}{21},
    \]
    which is false since $4 \cdot 21 = 84 > 69$, a contradiction.
\end{itemize}
\end{proof}

Finally, we prove Theorem~\ref{thm:satisfyability_clone_sharing}. The proof first expresses $D_k \circ Qw$ as a sum over equivalence classes, from which locality and each clone-sharing property follow by direct computation.

\begin{reptheorem}{thm:satisfyability_clone_sharing}
Let $w$ be a symmetric graph weighting function and $k \geq 1$.
Then $Qw$ is a $1$-clone-robust graph weighting function satisfying clone-sharing dilution. It fails $k$-clone-sharing positivity, and it fails $k$-clone-sharing domination whenever $w$ is positive.
Moreover, $D_k \circ Qw$ is a $2k$-clone-robust graph weighting function satisfying clone-sharing dilution and $k$-clone-sharing domination. If $w$ is positive, it furthermore satisfies $k$-strong clone-sharing dilution and $k$-strong clone-sharing domination.
\end{reptheorem}

\begin{proof}
We first state a helpful decomposition and then prove each point one after the other.

\begin{itemize}
    \item \textbf{Class decomposition.}
    For each class $C \in V(G)/{\equiv_G}$, we denote $W_C := w(G/{\equiv})(C)$ and $n_C := \vert N^k_G[u]\vert$ for any $u \in C$. 
    Grouping the terms of $D_k \circ Qw(G)(v)$ by classes and using $Qw(G)(u) = W_{[u]}/\vert [u]_G\vert$, each class $C \subseteq N^k_G[v]$ contributes $\vert C\vert \cdot W_C/(\vert C\vert\, n_C)$, so that
    \begin{equation}\label{eq:DkQ_classes}
        D_k \circ Qw(G)(v) = \sum_{C \subseteq N^k_G[v]} \frac{W_C}{n_C},
    \end{equation}
    where the sum ranges over the classes $C \in V(G)/{\equiv_G}$ contained in $N^k_G[v]$.

    \item \textbf{Well-definedness and symmetry.}
    Since $w(G/{\equiv})$ is a probability distribution, $\sum_{u \in V(G)} Qw(G)(u) = \sum_{C} \vert C\vert \cdot W_C/\vert C\vert = 1$. Since each vertex $u$ distributes its mass uniformly over $N^k_G[u]$, $D_k$ maps probability distributions to probability distributions. Any graph isomorphism $\sigma$ maps equivalence classes to equivalence classes, induces an isomorphism of the quotient graphs, and preserves $k$-neighborhoods; hence both $Qw$ and $D_k\circ Qw$ inherit the symmetry of $w$.

    \item \textbf{Locality.}
    Let $y \equiv_G x$ with $y \neq x$, let $G' \coloneq G \setminus \{y\}$, and let $z \in V(G')$. Since $G$ and $G'$ have the same quotient graph, they also have the same weights $W_C$, and the same classes except for $[x]_{G'}$, which loses one element.
    If $d_G(x,z) > 1$, then $z \notin [x]_G$, and therefore $Qw(G')(z) = W_{[z]}/\vert [z]\vert = Qw(G)(z)$. Thus, $Qw$ is $1$-local.
    
    If $d_G(x,z) > 2k$, then $[x]_G \not\subseteq N^k_G[z]$, so that the classes contained in $N^k[z]$ are the same in $G$ and $G'$. Moreover, removing $y$ only changes $n_C$ for the classes $C \subseteq N^k_G[x]$, and no such class can be contained in $N^k_G[z]$, since it would yield $d_G(x,z) \leq 2k$. Hence, every term of Equation~\eqref{eq:DkQ_classes} is the same in $G$ and $G'$, and $D_k \circ Qw(G')(z) = D_k \circ Qw(G)(z)$, i.e., it is $2k$-local.

    \item \textbf{Clone-sharing coefficients.}
    Fix $x, y \in V(G)$ and $i \geq 0$, with $G_x^0 \coloneq G$. Note that  $G_x^i$ has the same quotient graph and weights $W_C$ as $G$, the class of $x$ has size $c + i$ where $c := \vert [x]_G\vert$, and the classes contained in $N^k[y]$ are the same as in $G$. Moreover, its $k$-neighborhood sizes are $n_C + i$ if $C \subseteq N^k_G[x]$, and $n_C$ otherwise. Therefore, we have
    \begin{equation}\label{eq:sharing_Q}
        \cchi{Qw}{G_x^i}{y \to x} =
        \begin{cases}
            W_{[x]} \Big( \dfrac{1}{c+i} - \dfrac{1}{c+i+1}\Big) & \text{if } y \equiv_G x, \\
            0 & \text{otherwise,}
        \end{cases}
    \end{equation}
    and, using that only the classes contained in both $N^k_G[y]$ and $N^k_G[x]$ contribute in Equation~\eqref{eq:DkQ_classes},
    \begin{equation}\label{eq:sharing_DkQ}
        \cchi{D_k\circ Qw}{G_x^i}{y \to x} = \sum_{C \subseteq N^k_G[x] \cap N^k_G[y]} \frac{W_C}{(n_C+i)(n_C+i+1)}.
    \end{equation}

    \item \textbf{Properties of $Qw$.}
    Note that the clone-sharing coefficients of $Qw$ in Equation~\eqref{eq:sharing_Q} are non-increasing in $i$, so $Qw$ satisfies clone-sharing dilution.
    To show that $Qw$ fails $k$-positivity (and $k$-domination for positive $w$), let $G = P_3$ with vertices $a - b - c$.
    Since $b \not\equiv_G a$ and $d_G(a,b) = 1 \leq k$, Equation~\eqref{eq:sharing_Q} gives $\cchi{Qw}{G}{b \to a} = 0$, which violates $k$-positivity.
    Moreover, let $G' \coloneq G_a^1$. Since $N^k_{G'}[b] = V(G')$ and $N^k_{G'}[a^1] = N^k_{G'}[a]$, we have $N^k_{G'}[a] \cap N^k_{G'}[a^1] = N^k_{G'}[a] = N^k_{G'}[a] \cap N^k_{G'}[b]$, so that $k$-domination, applied in both directions, requires $\cchi{Qw}{G'}{a^1 \to a} = \cchi{Qw}{G'}{b \to a}$. However, if $w$ is positive, Equation~\eqref{eq:sharing_Q} with $c = 2$ gives $\cchi{Qw}{G'}{a^1 \to a} =  \frac{W_{[a]}}{6} > 0 = \cchi{Qw}{G'}{b \to a}$.

    \item \textbf{Properties of $D_k\circ Qw$.}
    
    \emph{a) Dilution:} every term of Equation~\eqref{eq:sharing_DkQ} is non-negative and non-increasing in $i$, so $D_k\circ Qw$ satisfies clone-sharing dilution. If $w$ is positive and $d_G(x,v) \leq k$, then $[x]_G \subseteq N^k_G[x] \cap N^k_G[v]$, and the term of $C = [x]_G$, with $W_{[x]} > 0$, is strictly decreasing in $i$. Hence, $D_k\circ Qw$ satisfies $k$-strong dilution.

    \emph{b) Domination:} let $y, z \in V(G)$ such that $N^k_G[x] \cap N^k_G[z] \subseteq N^k_G[x] \cap N^k_G[y]$, and let $\Lambda :=(N^k_G[x] \cap N^k_G[y]) \setminus N^k_G[z]$. Since both intersections are unions of classes, Equation~\eqref{eq:sharing_DkQ} with $i=0$ yields
    \begin{equation}
        \begin{aligned}
            \cchi{D_k\circ Qw}{G}{y \to x} - \cchi{D_k\circ Qw}{G}{z \to x} 
            &= \sum_{C \subseteq \Lambda} \frac{W_C}{n_C(n_C+1)} \\
            &\geq\; 0,
        \end{aligned}
    \end{equation}
    so $D_k\circ Qw$ satisfies $k$-domination. Note that $D_k\circ Qw$ also satisfies $k$-strong domination if $w$ is positive, because the sum is empty and the difference null exactly when the two intersections coincide.
\end{itemize}
\end{proof}

\section{Proofs of Clone-Sharing Lemmas}\label{app:proof_lem}

This appendix proves the lemmas of Section~\ref{sec:graph_weighting_design} relating the clone-sharing properties to each other and to locality, namely Lemmas~\ref{lem:dilut_imply_pos_shar}, \ref{lem:incompa_strong_clone_domin} and~\ref{lem:dilut_indep_domi}, followed by Lemma~\ref{lem:D^k_incompat} on iterated diffusion. We also show that locality cannot be strengthened to the addition of arbitrary neighbors (Lemma~\ref{lem:strict_loc}).

We start with the relation between clone-sharing dilution and non-negativity. The forward direction is a telescoping argument; for the converse, we construct a weighting on a single connected component of small diameter, so that locality imposes no constraint inside it.

\begin{replemma}{lem:dilut_imply_pos_shar}
    Clone-sharing dilution (resp. $k$-strong dilution) implies clone-sharing non-negativity (resp. $k$-positivity). However, the converse does not hold: there exists a $2$-clone-robust graph weighting function that satisfies clone-sharing non-negativity but fails clone-sharing dilution.
\end{replemma}

\begin{proof}
We show the two directions one after the other.
\begin{itemize}
    \item \textbf{(Strong) dilution $\implies$ non-negativity (positivity).}
    Let $G$ be a graph and $x, y \in V(G)$. Set $G_x^0 \coloneq G$ and write $w_i \coloneq w(G_x^i)(y)$ and $\chi_i \coloneq \cchi{w}{G_x^i}{y \to x} = w_i - w_{i+1}$ for $i \geq 0$, where the second equality uses $(G_x^i)_x^1 = G_x^{i+1}$. Summing these differences yields $w_N = w_0 - \sum_{i=0}^{N-1} \chi_i$ for every $N \geq 1$.

    Assume first that $w$ satisfies clone-sharing dilution. Then $\chi_i \leq \chi_0$ for all $i \geq 0$, so that $1 \geq w_N \geq w_0 - N\chi_0 \geq -N\chi_0$ for every $N \geq 1$. Letting $N \to \infty$ forces $\chi_0 \geq 0$.

    Assume now that $w$ satisfies $k$-strong clone-sharing dilution, which in particular implies dilution, and that $d_G(x,y) \leq k$. Since adding clones of $x$ preserves $d_G(x,y)$, strong dilution gives $\chi_0 > \chi_1$, while the above argument applied to $G_x^1$ in place of $G$ gives $\chi_1 \geq 0$. Hence $\chi_0 > 0$.

    \item    \textbf{The converse does not hold.}
    Let $\mathcal{C}$ denote the set of connected graphs of diameter at most $2$. Adding or removing a clone preserves graph distances between the remaining vertices, so $G \in \mathcal{C}$ if and only if $G_x^1 \in \mathcal{C}$, for every $x \in V(G)$. For $G \in \mathcal{G}$ and $v \in V(G)$, let $c_v \coloneq \vert [v]_G \vert$ denote the size of the equivalence class of $v$, and define
    \begin{equation*}
        w(G)(v) :=
        \begin{cases}
            \dfrac{1 + 1/c_v^2}{M(G)} & \text{if } G \in \mathcal{C}, \\[2mm]
            w^{\mathrm{CU}}(G)(v) & \text{otherwise,}
        \end{cases}
    \end{equation*}
    where $M(G) \coloneq \sum_{v \in V(G)} \Big(1 + \frac{1}{c_v^2}\Big) = \sum_{C \in V(G)/{\equiv_G}} \Big(\vert C\vert + \frac{1}{\vert C\vert}\Big).$

    \emph{a) $w$ is a $2$-clone-robust graph weighting function.} 
    On $\mathcal{C}$, the weights are positive and sum to $1$ by definition of $M(G).$
    Moreover, they only depend on class sizes, which are preserved by graph isomorphisms; outside $\mathcal{C}$, $w = w^{\mathrm{CU}}$ is symmetric.
    
    For $2$-locality, let $y \equiv_G x$ with $y \neq x$, and $z$ with $d_G(x,z) > 2$. As $G$ and $G \setminus \{y\}$ lie on the same side of the case distinction, either $G \in \mathcal{C}$ and no such $z$ exists, or $w(G)(z) = w(G\setminus\{y\})(z)$ by $1$-locality of $w^{\mathrm{CU}}.$

    \emph{b) $w$ satisfies clone-sharing non-negativity.} If $G \notin \mathcal{C}$, this follows from Example~\ref{exple:clone_shar_prop_w^CU}. If $G \in \mathcal{C}$, adding a clone of $x$ only increases $c_x$ by one, so that
    \begin{equation*}
        \begin{aligned}
            M(G_x^1) - M(G) &= \Big(c_x + 1 + \frac{1}{c_x+1}\Big) - \Big(c_x + \frac{1}{c_x}\Big) \\
            &= 1 - \frac{1}{c_x(c_x+1)} \geq \frac12 > 0.
        \end{aligned}
    \end{equation*}
    Hence, the denominator of $w(\cdot)(y)$ strictly increases, while its numerator stays the same if $y \not\equiv_G x$ and strictly decreases if $y \equiv_G x$. In both cases, $\cchi{w}{G}{y \to x} > 0$.

    \emph{c) $w$ fails clone-sharing dilution.} Let $H \coloneq P_3 \in \mathcal{C}$ with vertices $a - b - c$. The graphs $H$, $H_b^1$ and $H_b^2$ have class sizes $(1,1,1)$, $(1,2,1)$ and $(1,3,1)$ respectively, so that $M(H) = 6$, $M(H_b^1) = \frac{13}{2}$, $M(H_b^2) = \frac{22}{3}$, and
    $$w(H)(a) = \tfrac{1}{3}, \qquad w(H_b^1)(a) = \tfrac{4}{13}, \qquad w(H_b^2)(a) = \tfrac{3}{11}.$$
    Therefore, we have
    $$\cchi{w}{H}{a \to b} = \frac13 - \frac{4}{13} = \frac{1}{39} \;<\; \frac{5}{143} = \frac{4}{13} - \frac{3}{11} = \cchi{w}{H_b^1}{a \to b},$$
    which violates clone-sharing dilution.
    
\end{itemize}

\end{proof}

We next show that strong clone-sharing domination forces a large locality radius. A single path suffices: two vertices far from $x$ must share nothing with it by locality, while their neighborhoods intersect that of $x$ differently.

\begin{replemma}{lem:incompa_strong_clone_domin}
    $k$-strong clone-sharing domination is not compatible with $j$-locality, for any $j \leq 2k-1$.
\end{replemma}

\begin{proof}
    Since $j$-locality implies $(2k-1)$-locality for every $j \leq 2k-1$, it suffices to treat $j = 2k-1$. Suppose for contradiction that $w$ satisfies both $k$-strong clone-sharing domination and $(2k-1)$-locality.
    Consider the path $G \coloneq P_{2k+2}$ with vertices $0, 1, \dots, 2k+1$, and let $x \coloneq 0$, $u \coloneq 2k$ and $v \coloneq 2k+1$.
    On one hand, we have $\cchi{w}{G}{u \to x} = 0 = \cchi{w}{G}{v \to x}$ since $d_G(x,u) = 2k$ and $d_G(x,v) = 2k+1$ both exceed $2k-1.$
    On the other hand, $N_G^k[x] = \{0, \dots, k\}$, $N_G^k[u] = \{k, \dots, 2k+1\}$ and $N_G^k[v] = \{k+1, \dots, 2k+1\}$, so that
    $$N_G^k[x] \cap N_G^k[v] = \varnothing \;\subsetneq\; \{k\} = N_G^k[x] \cap N_G^k[u].$$
    By $k$-strong clone-sharing domination, the strict inclusion yields $\cchi{w}{G}{u \to x} > \cchi{w}{G}{v \to x}$, which contradicts the previous equality.
\end{proof}

The next lemma separates dilution and domination. One direction follows directly from Theorem~\ref{thm:satisfyability_clone_sharing} applied to $w^{\mathrm{U}}$; for the other, we construct an explicit $2k$-clone-robust graph weighting function.

\begin{replemma}{lem:dilut_indep_domi}
    Clone-sharing dilution and clone-sharing domination do not imply each other.
    Specifically, $w^{\mathrm{CU}}$ is $k$-clone-robust and satisfies clone-sharing dilution, but fails $k$-clone-sharing domination, for every $k>0$.
    Conversely, for every $k > 0$, there exists a $2k$-clone-robust graph weighting function that satisfies $k$-clone-sharing domination but fails clone-sharing dilution.
\end{replemma}

\begin{proof}

    To see that dilution does not imply domination, apply Theorem~\ref{thm:satisfyability_clone_sharing} to the uniform graph weighting function $w^{\mathrm{U}}$, which is positive: $w^{\mathrm{CU}} = Qw^{\mathrm{U}}$ is $1$-clone-robust, hence $k$-clone-robust for every $k>0$, satisfies clone-sharing dilution, and fails $k$-clone-sharing domination.

    We thus focus on proving that domination does not imply dilution.
    Fix $k > 0$. We call a vertex $v \in V(G)$ \emph{$k$-central} if $N^k_G[v] = V(G)$, and denote by $Z_k(G)$ the set of $k$-central vertices of $G$ and by $m \coloneq \vert Z_k(G)\vert$ its cardinality. Let $\mathcal{C}_k$ be the set of connected graphs of diameter at most $2k$.
    Adding or removing a clone preserves all graph distances between the remaining vertices, and the clone $x^1$ of $x$ satisfies $N^k[x^1] = N^k[x]$. Hence, $G \in \mathcal{C}_k$ if and only if $G_x^1 \in \mathcal{C}_k$, and $Z_k(G_x^1) = Z_k(G) \cup \{x^1\}$ if $x \in Z_k(G)$, while $Z_k(G_x^1) = Z_k(G)$ otherwise.

    \begin{itemize}

        \item \textbf{Construction.}
        Let $\eta := \frac{1}{2(k+2)}$ and consider the non-increasing function $h(\beta) := \eta\, \mathbb{1}[\beta \leq k+2]$, satisfying $\beta h(\beta) \leq \frac12$. For every graph $G$ and $v \in V(G)$, we define
        \[
            w(G)(v) :=
            \begin{cases}
                \dfrac{1 - m\,h(m)}{n} + h(m)\,\mathbb{1}[v \in Z_k(G)] & \text{if } G \in \mathcal{C}_k, \\[2mm]
                D_k \circ Q w^{\mathrm{U}}(G)(v) & \text{otherwise,}
            \end{cases}
        \]
        where $n := \vert V(G)\vert$ is the number of vertices in $G$ and  $m := \vert Z_k(G)\vert$ its number of $k$-central vertices.
        In words, on $\mathcal{C}_k$ every vertex receives the same base weight, and $k$-central vertices receive an additional bonus as long as there are at most $k+2$ of them.

        Note that this is a well-defined graph weighting function: on $\mathcal{C}_k$, the weights sum to $(1 - m\,h(m)) + m\,h(m) = 1$, and they are positive since $m\,h(m) \leq (k+2)\,\eta = \frac12$. Outside of $\mathcal{C}_k$, Theorem~\ref{thm:satisfyability_clone_sharing} ensures that $D_k \circ Qw^{\mathrm{U}}$ is a graph weighting function.

        \item \textbf{$w$ is $2k$-clone-robust.}
        On $\mathcal{C}_k$, the weights only depend on $n$, $m$ and on whether $v$ is $k$-central, all of which are invariant under graph isomorphisms; outside $\mathcal{C}_k$, the weighting function coincides with the symmetric function $D_k\circ Qw^{\mathrm{U}}$. Hence $w$ is symmetric.
        For $2k$-locality, let $y \equiv_G x$ with $y \neq x$, and $z$ with $d_G(x,z) > 2k$. Since $G$ and $G \setminus\{y\}$ lie on the same side of the case distinction, either $G \in \mathcal{C}_k$ and no such $z$ exists, or $w(G)(z) = w(G\setminus\{y\})(z)$ follows from the $2k$-locality of $D_k\circ Qw^{\mathrm{U}}$ (Theorem~\ref{thm:satisfyability_clone_sharing}).

        \item \textbf{$w$ satisfies $k$-clone-sharing domination.}
        Let $G \in \mathcal{G}$, and let $x, y, z \in V(G)$ such that $N^k_G[x] \cap N^k_G[z] \subseteq N^k_G[x] \cap N^k_G[y]$. If $G \notin \mathcal{C}_k$, then $\cchi{w}{G}{y \to x} \geq \cchi{w}{G}{z \to x}$ follows from Theorem~\ref{thm:satisfyability_clone_sharing}, so assume $G \in \mathcal{C}_k$.
        The graph $G_x^1$ has $n+1$ vertices and $m' $ $k$-central vertices, where $m' = m+1$ if $x \in Z_k(G)$ and $m' = m$ otherwise, and $Z_k(G_x^1) \cap V(G) = Z_k(G)$. Hence, for every $v \in V(G)$,
        \begin{equation*}
        \begin{aligned}
            \cchi{w}{G}{v \to x} &= \frac{1 - m\,h(m)}{n} - \frac{1 - m'\,h(m')}{n+1}, \\
            &+\; \mathbb{1}[v \in Z_k(G)] \,\big(h(m) - h(m')\big),
        \end{aligned}
        \end{equation*}
        where the first term is independent of $v$ and $h(m) - h(m') \geq 0$ since $h$ is non-increasing. 
        Therefore, subtracting the expressions for $v = y$ and $v = z$, the first term cancels and
        \begin{equation*}
        \begin{aligned}
            \cchi{w}{G}{y \to x} - \cchi{w}{G}{z \to x} 
            &= \big(\mathbb{1}[y \in Z_k(G)] - \mathbb{1}[z \in Z_k(G)]\big)\\
            &\times \big(h(m) - h(m')\big).
        \end{aligned}
        \end{equation*}
        If $x \notin Z_k(G)$, then $m' = m$ and the difference vanishes. Otherwise, $x \in Z_k(G)$ gives $N^k_G[x] = V(G)$, so that the hypothesis reads $N^k_G[z] \subseteq N^k_G[y]$. If $z \in Z_k(G)$, then $V(G) = N^k_G[z] \subseteq N^k_G[y]$, so $y \in Z_k(G)$ and the difference vanishes. If $z \notin Z_k(G)$, then the first factor is non-negative, and so is the difference. In all cases, $\cchi{w}{G}{y \to x} \geq \cchi{w}{G}{z \to x}$.

        \item \textbf{$w$ fails clone-sharing dilution.}
        Let $G = P_{k+2} \in \mathcal{C}_k$ with vertices $0, \dots, k+1$, and let $x = 1$. Vertices $1, \dots, k$ are $k$-central, so that $G_x^i$ has $n = k+2+i$ vertices and $m = k+i$ $k$-central vertices, including $x$. Using $1 - (k+1)\eta = \frac{k+3}{2(k+2)}$ and $1 - (k+2)\eta = \frac12$, we get
        \begin{equation*}
        \begin{aligned}
            w(G_x^1)(x) &= \frac{1}{2(k+2)} + \eta = \frac{1}{k+2},\\
            w(G_x^2)(x) &= \frac{1}{2(k+4)} + \frac{1}{2(k+2)},\\
            w(G_x^3)(x) &= \frac{1}{k+5},
        \end{aligned}
        \end{equation*}
            
        and therefore
        \begin{equation*}
            \begin{aligned}
                \cchi{w}{G_x^2}{x \to x} - \cchi{w}{G_x^1}{x \to x}
            &= \Big(\frac{k+3}{(k+2)(k+4)} - \frac{1}{k+5}\Big) - \frac{1}{(k+2)(k+4)},\\
            &= \frac{1}{(k+4)(k+5)} > 0,
            \end{aligned}
        \end{equation*}
            
        which violates clone-sharing dilution.
    \end{itemize}
\end{proof}

We next show that 1-locality cannot be meaningfully strengthened: requiring for example that  the addition of neighbors---rather than only equivalent nodes---leaves the weights of distant vertices unchanged contradicts symmetry.

\begin{lemma}\label{lem:strict_loc}
No graph weighting function $w$ can simultaneously satisfy \textbf{symmetry}, and the following 
\textbf{strong locality} property: for every graph 
$G \in \mathcal{G}$ and vertices $x,y,z \in V(G)$ with $y \in N_G[x]$ and $z \notin N_G[x]$, we have
$w(G)(z) \;=\; w(G \setminus \! \{y\})(z).$
\end{lemma}

\begin{proof}
Consider the graph $G$ consisting of a central vertex $d$ connected to three disjoint paths of length three. 
Explicitly, for each $i=1,2,3$, consider the path $P_i:
a_i - b_i - c_i$ and connect its endpoint $c_i$ to the central vertex $d.$

Now consider the subgraph $G\setminus \!( P_2 \cup P3).$ By symmetry, vertices $a_1$ and $d$ must receive the same non-negative weight, say $w_1$,  
and vertices $b_1$ and $c_1$ must receive the same weight, say $w_2$. We moreover have $2(w_1 + w_2) =1$ since $w(G\setminus \!( P_2 \cup P3))$ is a probability distribution.

We now add back vertex $c_2$: by strong locality, the weight of vertices $a_1$ and $b_1$ must remain unchanged, i.e., $w(G\setminus \! (P_3 \cup \{a_2,b_2\}))(a_1) = w_1$ and $w(G\setminus \! (P_3 \cup \{a_2,b_2\}))(b_1) = w_2.$ By symmetry, we moreover know that $w(G\setminus \! (P_3 \cup \{a_2,b_2\}))(c_2) = w(G\setminus \! (P_3 \cup \{a_2,b_2\}))(a_1)$ and $w(G\setminus \! (P_3 \cup \{a_2,b_2\}))(d) = w(G\setminus \! (P_3 \cup \{a_2,b_2\}))(b_1)$, an we conclude that $w(G\setminus \! (P_3 \cup \{a_2,b_2\}))(c_1) = 1 - 2(w_1+w_2) = 0$ since $w(G\setminus \! (P_3 \cup \{a_2,b_2\}))$ is a probability distribution.

We now iteratively add back $b_2$ and  $a_2$; the same reasoning holds and we get $w(G\setminus \! P_3) (c_1) = w(G\setminus \! P_3) (d)= w(G\setminus \! P_3) (c_2) = 0$, $w(G\setminus \! P_3) (a_1) = w(G\setminus \! P_3) (a_2)= w_1$ and $w(G\setminus \! P_3) (b_1) = w(G\setminus \! P_3) (b_2)= w_2.$ 

Finally, we iteratively add back vertices $c_3$, $b_3$ and $a_3$.
Importantly, the weights of $a_1, a_2,b_1,b_2$ cannot change under strong locality, and we conclude that all the newly introduced vertices must receive weight zero. We finally conclude using symmetry that $w_1 = w(G)(a_1) = w(G)(a_3) = 0$ and $w_2 = w(G)(b_1) = w(G)(b_3) = 0$: this is a contradiction.

\end{proof}

Finally, we prove Lemma~\ref{lem:D^k_incompat}, which justifies the use of $D_k$ rather than iterated one-hop diffusion in Theorem~\ref{thm:satisfyability_clone_sharing}. We take $w = w^{\mathrm{U}}$, so that $Qw = w^{\mathrm{CU}}$, and exhibit explicit graphs on which each property fails.

\begin{replemma}{lem:D^k_incompat}
    There exists a symmetric graph weighting function $w$ such that, for every $\ell \geq 2$ and $k \geq 1$, the graph weighting function $D^\ell \circ Qw$ satisfies neither clone-sharing dilution nor $k$-clone-sharing domination.
\end{replemma}

\begin{proof}
    We consider the uniform graph weighting function $w^{\mathrm{U}}$, denote $\hat w := D^\ell \circ Qw^{\mathrm{U}}$, and note that $Qw^{\mathrm{U}} = w^{\mathrm{CU}}$.
    \begin{itemize}
        \item \textbf{Reduction.}
        It suffices to exhibit a graph $G$ and a vertex $x$ with $\cchi{\hat w}{G}{x \to x} < 0$. Indeed, this violates clone-sharing non-negativity, hence dilution by Lemma~\ref{lem:dilut_imply_pos_shar}. Moreover, since $N^k_G[x] \cap N^k_G[z] \subseteq N^k_G[x] = N^k_G[x] \cap N^k_G[x]$ for every $z \in V(G)$, $k$-domination would imply $\cchi{\hat w}{G}{x \to x} \geq \cchi{\hat w}{G}{z \to x}$ for all $z$, and summing over $z \in V(G)$ would give
        \[
            \vert V(G)\vert \cdot \cchi{\hat w}{G}{x \to x} \;\geq\; \sum_{z \in V(G)} \cchi{\hat w}{G}{z \to x} = 1 - \big(1 - \hat w(G_x^1)(x^1)\big) \geq 0.
        \]
    
        \item \textbf{Construction.}
        For $t \geq 1$, let $G$ be the tree in which $x$ is adjacent to $b_1$ and $b_2$, and each $b_j$ is adjacent to $t$ leaves $c_{j,1}, \dots, c_{j,t}$:
        \[
        \begin{tikzpicture}[scale=0.8, every node/.style={circle,thick,draw,inner sep=1pt,minimum size=6mm}]
            \node (x) at (0,1.4) {$x$};
            \node (b1) at (-2,0.4) {$b_1$}; \node (b2) at (2,0.4) {$b_2$};
            \node (c11) at (-3,-0.8) {\scriptsize $c_{1,1}$}; \node (c1t) at (-1,-0.8) {\scriptsize $c_{1,t}$};
            \node (c21) at (1,-0.8) {\scriptsize $c_{2,1}$}; \node (c2t) at (3,-0.8) {\scriptsize $c_{2,t}$};
            \node[draw=none] at (-2,-0.8) {$\dots$}; \node[draw=none] at (2,-0.8) {$\dots$};
            \draw[thick] (x)--(b1) (x)--(b2) (b1)--(c11) (b1)--(c1t) (b2)--(c21) (b2)--(c2t);
        \end{tikzpicture}
        \]
        All equivalence classes of $G$ are singletons, so for every $i \geq 0$, $G_x^i$ has $N \coloneq 3 + 2t$ classes: $X \coloneq \{x, x^1, \dots, x^i\}$ and singletons. Hence $Qw^{\mathrm{U}}(G_x^i)$ assigns $\frac{1}{N(i+1)}$ to each vertex of $X$ and $\frac1N$ to every other vertex.
        Let $B :=\{b_1, b_2\}$ and $C := \{c_{j,m} : j \in \{1,2\},\, m \in \{1,\dots,t\}\}$. By symmetry, the lazy random walk on $G_x^i$ projected onto the types $X, B, C$ is a Markov chain with transition matrix
        \[
            T \coloneq
            \begin{pmatrix}
                \theta_i & 1-\theta_i & 0 \\[1mm]
                \frac{i+1}{i+t+2} & \frac{1}{i+t+2} & \frac{t}{i+t+2} \\[1mm]
                0 & \frac12 & \frac12
            \end{pmatrix},
            \qquad \theta_i \coloneq \frac{i+1}{i+3},
        \]
        with rows and columns ordered as $X, B, C$, and an $\ell$-step walk ending in $X$ ends at $x$ with probability $\frac{1}{i+1}$. Therefore,
        \begin{equation}\label{eq:what_star}
            \hat w(G_x^i)(x) = \frac{1}{N(i+1)} \Big( (T^\ell)_{XX} + 2\,(T^\ell)_{BX} + 2t\,(T^\ell)_{CX} \Big).
        \end{equation}

        \item \textbf{Asymptotics.}
        Fix $i$ and let $t \to \infty$. In the limit, a walk starting from $\{B, C\}$ never enters $X$, and its transition matrix $R$ (the restriction of $T$ to $\{B,C\}$) converges to
        $$R_\infty \coloneq
            \begin{pmatrix} 0 & 1 \\ \frac12 & \frac12 \end{pmatrix}.$$
        We let $\rho_m \coloneq (R_\infty^m)_{CB} \geq 0$
        denote the probability that this limiting walk, started in $C$, is in $B$ after $m$ steps. Since $X$ can only be entered from $B$, decomposing walks from $C$ according to their first entrance into $X$ gives
        $$(T^\ell)_{CX} = \sum_{m=0}^{\ell-1} (R^m)_{CB}\, T_{BX}\, (T^{\ell-1-m})_{XX}.$$
        Using $(R^m)_{CB} \xrightarrow[t\to\infty]{} \rho_m$, with $t \cdot T_{BX} = \frac{t(i+1)}{i+t+2} \xrightarrow[t\to\infty]{}  i+1$, and $(T^n)_{XX} \xrightarrow[t\to\infty]{}  \theta_i^n$ for every $n$, we get
        \begin{align*}
            (T^\ell)_{XX} &\to \theta_i^\ell, \qquad (T^\ell)_{BX} \to 0,\\
            t\,(T^\ell)_{CX} &\to (i+1) \sum_{m=0}^{\ell-1} \rho_m\, \theta_i^{\ell-1-m}.
        \end{align*}
        Multiplying Equation~\eqref{eq:what_star} by $2t$ and using $2t/N \to 1$, we conclude that
        \begin{equation*}
        \begin{aligned}
            L(i) &:= \lim_{t\to\infty} 2t\cdot \hat w(G_x^i)(x) \\
            &= \frac{1}{i+1}\Big( \theta_i^\ell + 2(i+1) \sum_{m=0}^{\ell-1} \rho_m\, \theta_i^{\ell-1-m} \Big)\\
            &= \frac{\theta_i^\ell}{i+1} + 2 \sum_{m=0}^{\ell-1} \rho_m\, \theta_i^{\ell-1-m}.
        \end{aligned}
        \end{equation*}

        \item \textbf{Conclusion.}
        Since $\lim_{t \to \infty} 2t \cdot \cchi{\hat w}{G}{x \to x} = L(0) - L(1)$, it suffices to show that $L(0) < L(1)$: then $\cchi{\hat w}{G}{x \to x} < 0$ for $t$ large enough, and the reduction step concludes the proof.
        
        We compare the two first terms of $L(i)$, using $\theta_0 = \frac13$ and $\theta_1 = \frac12$. For the first term, $\theta_0^\ell = \dfrac{1}{3^\ell} < \dfrac{1}{2^{\ell+1}} = \dfrac{\theta_1^\ell}{2}$, as $(3/2)^\ell > 2$ for $\ell \geq 2$. 
        For the second term, the sum is a polynomial in $\theta_i$ with non-negative coefficients $\rho_m$, hence it is non-decreasing in $\theta_i$, and $\theta_0 < \theta_1$.
    \end{itemize}
\end{proof}

\section{Sharing coefficient for the family $f_\nu$}\label{app:shar_coeff_f_nu}

All the above constructions extend naturally to the second family of weighting functions introduced in~\cite{berriaud_clone-robust_2025}.
Given a probability density $\nu$ supported on $[0,\alpha]$, this family is defined by
$$ f_\nu(S)(x) = \int_0^{\alpha} \nu(r)\, g_r(S)(x)\,dr.$$
We extend the sharing coefficient of $g_r$ to $f_\nu$ by similarly integrating $\chi^{g_r}$ over all radii.

\begin{definition}[Sharing coefficient of $f_\nu$]\label{def:shar_coeff_fnu}
    Let $S \subseteq \mathbb{R}^n$ be finite, and let $x,y \in S$. The sharing coefficient of $f_\nu$ is defined by
    \begin{equation*}
        \chi^{f_\nu,S}_{y \to x} \coloneq \int_0^{\alpha} \nu(r)\, \chi^{g_r,S}_{y\to x}\,dr,
    \end{equation*}
    which, for $y = x$, defines the private weight $\chi^{f_\nu,S}_{x \to x}$ of $x$.
\end{definition}

Both quantities are well defined, since $0 \leq \chi^{g_r,S}_{y\to x} \leq g_r(S)(x) \leq 1$ for every $r$, so that the integrand is bounded.
Moreover, most qualitative properties of $\chi^{g_r}$ directly carry over to $\chi^{f_\nu}$:
\begin{itemize}
    \item \emph{Non-negativity and locality.} $\chi^{f_\nu,S}_{y \to x} \geq 0$, with equality if and only if $d(x,y) \geq 2\alpha$, which is the locality parameter of $f_\nu$. Indeed, $\chi^{g_r,S}_{y \to x}$ vanishes exactly for $r \leq d(x,y)/2$, and $\nu$ is supported on $[0,\alpha]$.
    \item \emph{Decomposition.} Integrating the definition of the private weight of $g_r$, $f_\nu(S)(x)$ decomposes into private and shared parts, i.e.,
    $$ f_\nu(S)(x) = \chi^{f_\nu,S}_{x \to x} + \sum_{y\neq x} \chi^{f_\nu,S}_{y \to x}.$$
    \item \emph{Domination.} If $y$ intersection-dominates $z$ with respect to $x$ for every $r \in [0,\alpha]$, then Property~\ref{lem:sharing_domination_gr} holds for every radius, and integrating yields $\chi^{f_\nu,S}_{y \to x} \geq \chi^{f_\nu,S}_{z \to x}$.
    \item \emph{Dilution.} Similarly, integrating Property~\ref{lem:sharing_dilution_gr} over $r$ shows that $\chi^{f_\nu,S_i}_{y \to x}$ is non-increasing as perfect clones of $x$ are added.
\end{itemize}

The only property that does not directly extend to $f_\nu$ is the simple multiplicative form of the element-removal effect.
Integrating Equation~\eqref{equ:gr_removal_effect} with respect to $\nu$ gives
\begin{equation*}
    f_\nu(S\setminus \! \{x\})(y)
    = \int_0^\alpha \nu(r)\, \frac{g_r(S)(y) + \chi^{g_r,S}_{y\to x}}{1-\chi^{g_r,S}_{x\to x}}\,dr
    \;\geq\; f_\nu(S)(y) + \chi^{f_\nu,S}_{y \to x},
\end{equation*}
where the inequality uses $1 - \chi^{g_r,S}_{x\to x} \in (0,1]$.
Thus, removing an element $x$ still induces a \emph{local additive increase} of at least $\chi^{f_\nu,S}_{y \to x}$ in the weight of each remaining element $y$, and in particular never decreases it. However, since the rescaling factor $(1-\chi^{g_r,S}_{x\to x})^{-1}$ depends on $r$, the global rescaling is in general no longer purely multiplicative, i.e., $f_\nu(S\setminus \! \{x\})(y)$ differs from $\big(f_\nu(S)(y) + \chi^{f_\nu,S}_{y \to x}\big)\big(1-\chi^{f_\nu,S}_{x\to x}\big)^{-1}$ unless $\chi^{g_r,S}_{x\to x}$ is constant across radii.

\section{Feasibility of the Edge-Sharing Constraints}\label{app:feasability}

We complement Theorem~\ref{thm:impossible_edge_sharing} by testing computationally whether its weaker variant, which replaces endpoint expansion by non-negative endpoint sharing, $\overset{\mathrm{e}}{\chi}{}^{w,G}_{u \to e} \geq 0$ for $u \in e$, is compatible with clone-robustness: each endpoint may then scale by its own factor, which need only be at least the global factor $\eta_{G,e}$.
That is, we look for a $1$-clone-robust graph weighting function $w$ such that, for every graph $G$ and edge $e \in E(G)$, there exists $\eta_{G,e} \in (0,1]$ with
\begin{equation*}
    \begin{aligned}
        w(G-e)(z) &\geq \eta_{G,e}\, w(G)(z) \quad \text{ for } z \in e, \\
        w(G-e)(z) &\leq \eta_{G,e}\, w(G)(z) \quad \text{ for } z \notin e,
    \end{aligned}
\end{equation*}
with equality whenever $d_G(e,z) > 1$.

\paragraph{Formulation.}
Restricting to all graphs on at most $n$ vertices, we introduce one variable $w_{G,v} \in [0,1]$ for each graph $G$, up to isomorphism (using a canonical labelling), and each vertex $v \in V(G)$, and one variable $\lambda_{G,e} \coloneq 1/\eta_{G,e} \geq 1$ for each edge $e \in E(G)$.
The constraints are:
(i)~normalization, $\sum_{v \in V(G)} w_{G,v} = 1$ for every $G$;
(ii)~symmetry, $w_{G,u} = w_{G,\sigma(u)}$ for every automorphism $\sigma$ of $G$ and $u \in V(G)$;
(iii)~$1$-locality, $w_{G',z} = w_{G,z}$ whenever $G'$ is obtained from $G$ by adding a clone of a vertex $x$ and $z \notin N_G[x]$;
(iv)~the edge-sharing constraints above, rewritten as $\lambda_{G,e}\, w_{G-e,z} - w_{G,z} \geq 0$ for $z \in e$, $\leq 0$ for $z \notin e$, and $= 0$ whenever $d_G(e,z) > 1$.
The constraints in~(iv) are bilinear in $(\lambda, w)$, so the resulting problem is a non-convex quadratically constrained feasibility program.
As a sanity check of the encoding, additionally requiring a common expansion factor $\mu_{G,e} \geq 1$ for both endpoints, as in Theorem~\ref{thm:impossible_edge_sharing}, makes the program infeasible already for $n = 5$.

\paragraph{Implementation.}
A first encoding using the SMT solver \emph{Z3}, through Python and then Rust, did not terminate within a day for $n = 5$.
We therefore switched to the numerical solver \emph{Gurobi}, which handles non-convex quadratic constraints by spatial branch-and-bound and thus provides global feasibility and infeasibility guarantees, up to numerical tolerances.
A Rust program enumerates all graphs up to isomorphism and writes the program in LP file format, which is then passed to the solver.

\paragraph{Results.}
Gurobi finds a feasible assignment for all $208$ graphs on at most $n = 6$ vertices, but reports the program infeasible for $n = 7$ ($1{,}252$ graphs, about $21{,}000$ variables and $85{,}000$ bilinear constraints), both with default settings (after $12$ minutes) and with its most conservative numerical settings (after $2.7$ hours on a personal computer).
In contrast, dropping the equality for vertices far from $e$ yields a program that is feasible for all graphs on at most $7$ vertices, so this locality requirement drives the contradiction.
This strongly suggests that Theorem~\ref{thm:impossible_edge_sharing} still holds under non-negative endpoint sharing, although extracting a human-checkable proof is left for future work.
The code is available \href{https://anonymous.4open.science/r/edge_sharing_solver-690D/README.md}{here}.

\section{Graph Weigthing Functions based on Maximal Clique Covers}\label{sec:MCC_weighting}

One approach to constructing neighborhood-based graph weighting functions relies on the following observation: the set of maximal cliques of a graph, denoted by $\mathcal{K}(G)$, is invariant under the removal of equivalent nodes.
Indeed, for any finite graph $G$ and two equivalent vertices $x \equiv_G y$, we have, i.e.,
$$\mathcal{K}(G\setminus\{x\}) = \{\,K \setminus\{x\} \mid K \in \mathcal{K}(G)\,\}.$$
Moreover, $\mathcal{K}(G)$ forms a cover of the vertex set $V(G)$, that is, 
$\bigcup_{K \in \mathcal{K}(G)} K = V(G)$. 
We refer to such a cover of $V(G)$ by maximal cliques as a \emph{maximal-clique-cover (MCC)}.

Starting from an MCC--say $\mathcal{K}(G)$--one can construct a clone-robust weighting function by assigning weight uniformly to each maximal clique and then distributing each clique’s weight among its constituent nodes. 
A key design question is how to handle vertices that belong to multiple maximal cliques. 
The most straightforward approach is to divide each clique’s weight uniformly among its vertices and compute each node’s total weight as the sum of its contributions across cliques.
Formally, we define the \emph{additive maximal-clique weighting rule} $w^{\mathrm{MCCA}}$ as follows, i.e.,
$$w^{\mathrm{MCCA}}(G)(v) = \frac{1}{\vert \mathcal{K}(G)\vert}
\sum_{K \in \mathcal{K}(G):\, v \in K} \frac{1}{ \vert K\vert},
\qquad v \in V(G).$$

Alternatively, one can introduce a notion of varying participation: 
each node contributes equally to all cliques it belongs to, and the total weight within each clique is then shared proportionally to the participation of a node within that clique. 
The resulting node weights are obtained by summing each node’s contributions across all cliques. 
Formally, we define for each vertex $v\in V(G)$ the number of maximal cliques it belongs to as $c_v := \vert \{K \in \mathcal{K}(G) \mid v \in K\}\vert .$ We then interpret $1/c_v$ as the fraction of participation of $v$ in each clique it belongs to, and compute the total participation within a clique $K$ as $P_K := \sum_{v \in K} \frac{1}{c_v}.$
Then each maximal clique gets total mass $1/\vert \mathcal{K}(G)\vert$ and redistributes it proportionally to the participation of its nodes, leading to
$$w^{\mathrm{MCCP}}(G)(v) =\frac{1}{\vert \mathcal{K}(G)\vert} \sum_{K \in \mathcal{K}(G):\, v \in K} \frac{1}{ c_v \cdot P_K}, \qquad v \in V(G).$$

Note that both $w^{\mathrm{MCCA}}$ and $w^{\mathrm{MCCP}}$ are clearly 1-clone-robust graph weighting functions: they verify symmetry because $\mathcal{K}(G)$ is invariant under graph isomorphism, and locality since the removal of a clone only affects the weights of nodes within the same cliques, i.e., its neighbors.

\end{document}